\documentclass[sip,biber]{now-journal} 

\usepackage[utf8]{inputenc}
\usepackage{amssymb}
\usepackage{float}

\usepackage{amsmath}
\usepackage{graphicx}
\usepackage{multirow}
\usepackage{threeparttable}
\usepackage{amssymb}
\usepackage{amsfonts}
\usepackage{algorithmic}
\usepackage{algorithm}
\usepackage{booktabs} 
\usepackage{subfigure}
\usepackage{biblatex}

\newcommand{\ie}{\emph{i.e.}}
\newcommand{\eg}{\emph{e.g.}}

\theoremstyle{plain}
\newtheorem{thm}{Theorem}[section]
\theoremstyle{remark}

\theoremstyle{remark}
\newtheorem*{rem*}{Remark}
\theoremstyle{plain}
\newtheorem{prop}[thm]{Proposition}
\newcommand{\blue}[1]{\textcolor{black}{#1}} 

\newif\ifdoubleblind

\title{Multimodal Signal Restoration with Signed Twofold Graph Learning}

\ifdoubleblind
  \author[1*]{Anonymous}
  \author[2]{Anonymous}
  \author[1]{Anonymous}
  \affil[1]{Anonymous Institution}
  \affil[2]{Anonymous Institution}
\else
  \author[1*]{Haruki Yokota}
  \author[2]{Hiroshi Higashi}
  \author[1]{Yuichi Tanaka}
  \affil[1]{The University of Osaka, Suita, Japan}
  \affil[2]{Kansai University, Suita, Japan}
\fi

\issuevolumeyear{2026}
\issuevolumenumber{xx}
\issueelocation{e} 
\articletype{Original Paper} 

\ifdoubleblind
  \articledatabox{Received xx xxxxx 2026; Revised xx xxxxx 2026\\[2pt]
  ISSN xxxx-xxxx; DOI 10.1561/xxx.xxxxxxxx}
\else
  \articledatabox{Received xx xxxxx 2026; Revised xx xxxxx 2026\\[2pt]
  ISSN xxxx-xxxx; DOI 10.1561/xxx.xxxxxxxx\\
  \copyright\ H. Yokota, H. Higashi and Y. Tanaka}
\fi

\OAstatement{This is an Open Access article, distributed under the terms of the Creative Commons Attribution licence (\url{http://creativecommons.org/licenses/by-nc/4.0/}), which permits unrestricted re-use, distribution, and reproduction in any medium, for non-commercial use, provided the original work is properly cited.}


\ifdoubleblind\else
  \creditline*{Corresponding author: h.yokota@sip.comm.eng.osaka-u.ac.jp. The preliminary version of this work was presented at the Asia Pacific Signal and Information Processing Association Annual Summit and Conference (APSIPA ASC) 2025 \cite{yokota2025unrolled}.}
\fi

\begin{document}

\begin{abstract}
\blue{We propose a method for jointly learning signed twofold graphs and performing signal restoration on multimodal graph signals.}
Multimodal signals on sensor networks are commonly modeled under the twofold graph assumption (TGA), which represents spatial structure and inter-modality relations as two separate graphs.
Existing TGA-based signal restoration methods, however, either assume the graphs are known or restrict edge weights to be non-negative, preventing them from capturing negative inter-modal correlations. 
\blue{We address both limitations as follows.
To learn twofold graphs from noisy and incomplete data,} we formulate joint signal restoration and twofold graph learning as MAP estimation under a matrix normal prior, where the spatial and modality graph Laplacians appear directly as precision matrices.
The resulting non-convex objective is solved by alternating minimization:
The signal is updated via conjugate gradient applied to the arising Sylvester-type linear system; the graphs are updated via primal-dual hybrid gradient (PDHG).
\blue{To capture negative inter-modal correlations, we} estimate the signed structure of the modality graph from the dominant eigenspace of a complementary kernel matrix, which is then used in PDHG to update edge magnitudes.
These iterative solvers are then unrolled into a feedforward network, with regularization weights and step sizes as layer-wise trainable parameters.
Experiments on synthetic multimodal graph signals and \blue{two real-world datasets (Japan meteorological and Beijing air-quality data)} confirm that the proposed method outperforms existing baselines across a range of noise levels and missing-data patterns\blue{; the learned graphs are also directly validated against the ground-truth graphs on the synthetic datasets}.
\end{abstract}

\section{Introduction}
In sensor networks, sensors often capture a rich variety of data simultaneously, like video/audio streams and temperature/pressure readings.
Such data are called multimodal signals.
In practice, real-world multimodal signals are imperfect, suffering from noise or missing values.
Signal restoration for those multimodal signals, including denoising and interpolation, is therefore crucial for a wide range of applications, such as autonomous driving, human-computer interaction, and medical diagnostics \cite{xiao2022multimodal, liu2024EEGBased,liu2024Perspectives,arceo2025Robust}. 

A key principle in signal restoration on sensor networks is utilizing underlying structures of data, whose structures are mathematically represented as graphs \cite{shuman_emerging_2013, ortega_graph_2018, tanaka_sampling_2020}. 
Multimodal signals on a network can be modeled as a spatial graph, with underlying spatial structure, and inter-modality relations are represented by a modality graph. This assumption is known as \textit{twofold graph assumption} (TGA) \cite{nagahama2022multimodal}.

Previous works for signal restoration with TGA face two key challenges.
First, many approaches either assume graphs are known a priori or rely on heavily parameterized networks to learn them \cite{shahid2016fastrobust,wang2020graphsampling,zheng2022multimodal}, limiting applicability in data-scarce scenarios.
Second, all existing TGA methods are designed for \textit{unsigned} graphs, modeling only pairwise \textit{similarities} \cite{kojima2026unrolling}.
This neglects \textit{negative correlations} across modalities---for example, in environmental monitoring, the anti-correlation between temperature and precipitation is naturally encoded as a negative edge weight. Such a relationship can be represented by signed graphs \cite{Dinesh2025, yokota2025efficient}, but unsigned graphs cannot.

In this paper, we address these limitations by proposing a method that:
1) formulates joint signal restoration and twofold graph learning as MAP estimation under a matrix normal prior, linking the graph Laplacians directly to the signal prior;
2) learns a spatial unsigned graph and a modality signed graph simultaneously, using a sign-magnitude decoupling with a spectral sign estimator; and
3) unrolls the resulting alternating iterative solvers into an interpretable deep network whose hyperparameters are trained end-to-end.

Specifically, we model the multimodal signal as a matrix normal random variable \cite{gupta1999Matrix} with graph Laplacians as the precision matrices, yielding a MAP objective that couples data fidelity with twofold graph smoothness.
We solve the resulting non-convex problem by alternating minimization over three variables: the signal is updated by solving a Sylvester-type linear system via conjugate gradient (CG); the spatial graph is updated by primal-dual hybrid gradient (PDHG) \cite{condat_primal_2013}; and the signed modality graph is updated by decoupling sign from magnitude, estimating the sign structure via subspace iteration and solving for edge magnitudes with PDHG.
Finally, we unroll the CG and PDHG iterations into a feedforward network \cite{monga2021algorithm}, with regularization weights and step sizes as layer-wise trainable parameters.

We evaluate the proposed method on synthetic multimodal graph signals and \blue{two real-world datasets: a Japan meteorological dataset and a Beijing air-quality dataset}.
Experimental results demonstrate that the proposed method outperforms competing baselines in both denoising and interpolation across a range of noise levels and missing-data patterns\blue{, with statistically significant differences}.
\blue{We also validate the learned graphs directly on the synthetic datasets, 
where the learned graphs estimate the ground-truth edges more accurately in terms of edge AUC and relative adjacency errors than those of competing graph learning methods in most cases.}

\blue{The rest of this paper is organized as follows.
Section~\ref{sec:related} reviews related work on multimodal graph signal restoration and unsigned/signed graph learning.
Section~\ref{sec:proposed} presents the proposed joint formulation, the CG- and PDHG-based solvers, and the unrolled architecture.
Section~\ref{sec:experiments} reports experiments on synthetic and real-world datasets.
Section~\ref{sec:conclusion} concludes the paper.}

\vspace{3pt}\textit{Notation:}
We use bold uppercase letters for matrices ($\mathbf{A}$), and bold lowercase letters for vectors ($\mathbf{a}$).
The corresponding non-bold letters with subscripts denote their elements ($A_{ij}$ and $a_i$).
The $j$th column of matrix $\mathbf{A}$ is denoted as the vector $\mathbf{a}_{j}$. The identity matrix is $\mathbf{I}$, and vectors of all ones and zeros are $\mathbf{1} = [1,\dots,1]^{\top}$ and $\mathbf{0} = [0,\dots,0]^{\top}$, respectively. The \blue{$\ell_p$-norm} is $\|\cdot\|_p$ and the Frobenius norm is denoted as $\|\cdot\|_F$.
An undirected graph is denoted as $\mathcal{G}$, and we denote the set of $N$ nodes in $\mathcal{G}$ as $\mathcal{V}=\{v_1,v_2,\dots,v_N\}$. The topology and edge weights of $\mathcal{G}$ are defined by a symmetric weighted adjacency matrix $\mathbf{W}\in \mathbb{R}^{N\times N}$, where $W_{ij}$ represents the weight of the edge between node $v_i$ and $v_j$.
Throughout the paper, we assume $W_{ii}=0$, i.e., the graphs have no self-loops. 
The degree matrix is the diagonal matrix denoted as $\mathbf{D} = \operatorname{diag}(d_1,\dots,d_N)$, where $d_i=\sum_{j} W_{ij}$. The combinatorial graph Laplacian $\mathbf{L}$ for unsigned graphs is given by $\mathbf{L}=\mathbf{D}-\mathbf{W}$, and those for signed graphs are defined as $\bar{\mathbf{L}}=\bar{\mathbf{D}}-\bar{\mathbf{W}}$, where $\bar{\mathbf{D}}$ is an absolute degree matrix whose $i$th diagonal element is defined as $\bar{d}_{ii}=\sum_j|\bar{W}_{ij}|$.

\section{Related Works}
\label{sec:related}
In this section, we review the three bodies of work most directly related to the proposed method: multimodal graph signal restoration, unsigned graph learning, and signed graph learning.

\subsection{Multimodal Graph Signal Restoration}
Under the twofold graph assumption (TGA) \cite{nagahama2022multimodal,kojima2026unrolling}, a multimodal signal $\mathbf{X}\in\mathbb{R}^{N\times M}$ is assumed to be smooth on a spatial graph $\mathcal{G}_s$ over $N$ nodes and on a modality graph $\mathcal{G}_m$ over $M$ modalities simultaneously.
Signal restoration under TGA is then formulated as minimizing a regularizer involving both graph Laplacians.

The conventional TGA-based methods only consider the sum of two graph regularizers \cite{nagahama2022multimodal,kojima2026unrolling};
\begin{equation}
\label{eq:additive_tga}
\min_\mathbf{X} \|\mathbf{Y}-\mathcal{A}(\mathbf{X})\|_F^2 + \alpha\text{tr}(\mathbf{X}^\top\mathbf{L}_s\mathbf{X}) + \beta\text{tr}(\mathbf{X}\mathbf{L}_m\mathbf{X}^\top),
\end{equation}
where $\mathcal{A}(\cdot)$ is a degradation operator, and $\alpha$ and $\beta$ are regularization parameters. This formulation does not explicitly model the interaction between the two graphs.
\blue{Here, the two graphs are assumed to be given a priori in \cite{nagahama2022multimodal}: This is a limitation in practical scenarios where the graphs are unknown.
LLAP-DAU \cite{kojima2026unrolling} removes this assumption by learning the twofold graph inside its unrolled network;
however, the learned graphs are restricted to be unsigned, and negative inter-modal correlations can only be represented as weak or missing edges.
In addition, each layer of LLAP-DAU updates the signal with closed-form graph low-pass filters that require a matrix inversion for each of the spatial and modality sub-problems, whose cost grows cubically with the number of nodes and modalities.
}

Another assumption based on the combinations of spatial and temporal graphs is known as \textit{product graphs} \cite{jimenez2018product,jiang2021sensor}. 
Product graphs allow the modeling of the mixture of two graph priors, which has been shown to improve the performance of posterior tasks such as denoising and interpolation of signals observed in a spatio-temporal environment. 

The matrix normal prior considered in this work is related to the Kronecker-product graph, where the variation of signals observed on a node of the spatial graph is dominated by the structure defined in the modality graph. 

Note that, product graph-based methods often simplify the temporal (modality) graph with a simple graph structure, such as a path graph \cite{jiang2021sensor}, which ignores the correlational information of the domain.
Also, the conventional product graph methods require a large product graph Laplacian matrix $\mathbf{L}_\text{product}\in \mathbb{R}^{NM\times NM}$, whose dimension is the product of the dimensions of the two graphs.

Methods based on graph neural networks \cite{shahid2016fastrobust,wang2020graphsampling,zheng2022multimodal} can learn the graph implicitly but require large training sets and offer limited interpretability.

\blue{Algorithm unrolling has also been studied for graph signal restorations.
In \cite{nagahama2022nested}, a graph signal restoration method with nested algorithm unrolling has been proposed, where an unrolled graph signal denoiser is embedded within an unrolled restoration algorithm, achieving interpretable restoration with a small number of parameters.
However, the graph is fixed and given, and inter-modality relations are not modeled.
In \cite{batreddy2025inpainting}, an unrolled alternating scheme for joint inpainting and graph learning has been proposed, which is close in spirit to our approach.
Nevertheless, it learns a single unsigned graph, and thus neither the twofold structure nor negative correlations are captured.}

\subsection{Unsigned Graph Learning}
Let us consider an observation model of a graph signal $\mathbf{y} \in \mathbb{R}^{N}$ on a single spatial graph $\mathcal{G}$ with $N$ nodes as $\mathbf{y} = \mathbf{x} + \boldsymbol{\epsilon}$
where $\mathbf{x}$ is the unknown true signal and $\boldsymbol{\epsilon}$ is additive white Gaussian noise (AWGN) with $\boldsymbol{\epsilon} \sim \mathcal{N} (\mathbf{0},\sigma^{2}_{\epsilon}\mathbf{I}_{N})$.

Based on the graph factor analysis model \cite{dong_learning_2016}, $\mathbf{x}$ can be seen as a sample drawn from a multivariate Gaussian distribution with the precision matrix $\mathbf{L}$: $\mathbf{x} \sim \mathcal{N}(\boldsymbol{\mu}_x,\mathbf{L}^{\dagger}+\sigma_{\epsilon}^2\mathbf{I}_{N})$
where $\boldsymbol{\mu}_x$ is the mean of $\mathbf{x}$ and $(\cdot)^\dagger$ represents the Moore-Penrose pseudo inverse.  

The joint graph learning and signal restoration under this model is discussed in \cite{dong_learning_2016}, where they formulate an optimization problem to find $\mathbf{x}$ and $\mathbf{L}$ from $\mathbf{y}$ as:
\begin{align}
\label{eq:dong_smooth}
\min_{\mathbf{x},\mathbf{L}\in\mathcal{L}} \|\mathbf{x} - \mathbf{y}\|_2^2 + \alpha \mathbf{x}^\top\mathbf{L}\mathbf{x} + \frac{\beta}{2}\|\mathbf{L}\|_F^{2}\end{align}
where $\mathcal{L}$ is a set of graph Laplacians given by 
\begin{equation}
\mathcal{L} = \{\mathbf{L}|\text{tr}(\mathbf{L})=N,\ L_{ij} = L_{ji} \leq0,\ i\neq j, \mathbf{L}\mathbf{1} = 0\}, \nonumber
\end{equation}
and $\alpha$ and $\beta$ are regularization parameters. The first term measures the data fidelity, and the second term measures the signal variation on the given $\mathbf{L}$ where $\mathbf{x}^\top\mathbf{L}\mathbf{x}=\sum_{i<j}^N W_{ij}(x_i - x_j)^2$, and the third term is a regularization on edge weights. 
\eqref{eq:dong_smooth} is generally non-convex due to the coupling of variables $\mathbf{x}$ and $\mathbf{L}$ in the second term.

In \cite{dong_learning_2016}, the following alternating optimization is applied.
\begin{enumerate}
    \item $\mathbf{L}$ is optimized with the fixed $\mathbf{x}$ by solving the following constrained quadratic optimization problem:
\begin{equation}
\label{eq:graph_learning}
\min_{\mathbf{L}\in\mathcal{L}} \alpha\mathbf{x}^\top \mathbf{Lx} + \beta\|\mathbf{L}\|_F^2.
\end{equation}
    \item $\mathbf{x}$ is obtained with the fixed $\mathbf{L}$ by solving 
\begin{equation}
\min_\mathbf{x} \|\mathbf{x}-\mathbf{y}\|_2^2 + \alpha\mathbf{x}^\top\mathbf{Lx}.
\end{equation}
This has a closed-form solution $\mathbf{x} = \left(\mathbf{I}_N+\alpha\mathbf{L}\right)^{-1} \mathbf{y}$ as long as $\left(\mathbf{I}_N+\alpha\mathbf{L}\right)$ is invertible.
\end{enumerate}
These two steps are iterated until the difference from the previous iteration reaches a pre-defined tolerance value.
\blue{While this framework is effective for a single graph, it considers a single unsigned graph for vector-valued signals:
Applying it to multimodal matrix data requires learning each graph independently, which ignores the coupling between the spatial and modality structures.}

\subsection{Signed Graph Learning}

Most graph learning methods, including the aforementioned one, are designed for unsigned graphs, i.e., $W_{ij} \ge 0$ for all $i,j$.
Unsigned graphs encode similarities.
However, it is natural that we often encounter dissimilarities as well: In this case, signed graph learning is required \cite{yokota2025efficient,yang_signed_2022,dinesh_lineartime_2022}.

A signed graph uses both positive and negative edge weights, which are stored in a \textit{signed adjacency matrix} $\bar{\mathbf{W}}\in\mathbb{R}^{N\times N}$. 
The goal is to find a $\bar{\mathbf{W}}$ that reflects the relationships in a set of observed signals \blue{$\mathbf{X}\in\mathbb{R}^{P\times N}$, where $P$} is the number of observations or features and $N$ is the number of nodes.

An edge weight of a signed graph can be represented as $\bar{W}_{ij} = S_{ij} |\bar{W}_{ij}|$ where $S_{ij} = \operatorname{sign}(\bar{W}_{ij})$ is the edge sign matrix.
Therefore, signed graph learning is formulated as follows \cite{dittrich_learning_2020}:
\begin{equation}
\label{eq:signed_GL}\min_{\mathbf{S},\bar{\mathbf{W}}} \alpha\sum_{i,j}|\bar{W}_{ij}|\|\mathbf{x}_i-S_{ij}\mathbf{x}_j\|_2^2 + \beta\|\bar{\mathbf{W}}\|_F^2
\end{equation}
where $\bar{\mathbf{W}}$ is a symmetric matrix with zero-diagonals, and $\mathbf{S}\in\{-1,1\}^{N\times N}$.
Optimizing $\mathbf{S}$ together with $\bar{\mathbf{W}}$ is non-convex and hence, a two-step approach is considered in \cite{dittrich_learning_2020}.

First, compute the distance matrix $\bar{\mathbf{Z}}$ together with the sign matrix $\mathbf{S}$ as:
\begin{align}
\label{eq:signed_distance}
\bar{Z}_{ij} &= \min\{\|\mathbf{x}_i - \mathbf{x}_j\|_{2}^{2} , \|\mathbf{x}_i + \mathbf{x}_j\|_{2}^{2}\}\nonumber \\
S_{ij} &= \begin{cases}1 & \text{if } \|\mathbf{x}_i - \mathbf{x}_j\|_{2}^{2} \leq \|\mathbf{x}_i + \mathbf{x}_j\|_{2}^{2} \\ -1 & \text{otherwise}.\end{cases}
\end{align}
Second, given these precomputed signs $\mathbf{S}$ and distances  $\bar{\mathbf{Z}}$, optimize the weight magnitude $\mathbf{W}\geq\mathbf{0}$ by solving a convex optimization problem;
\begin{align}
\label{eq:GL_signed}
\min_{\mathbf{W}\in\mathcal{W}}\ &\alpha\sum_{i,j}W_{ij}\bar{Z}_{ij} + \frac{\beta}{2}\|\mathbf{W}\|_F^2,
\end{align}
where $\mathcal{W}=\{\mathbf{W}|W_{ij}=W_{ji},\ W_{ij}\geq0,\ W_{ii} =0\ \forall i\}$.
The signed adjacency matrix is then given by $\bar{\mathbf{W}}^\star = \mathbf{S} \circ \mathbf{W}$, where $\circ$ is the element-wise product. The minimization in \eqref{eq:GL_signed} can be efficiently solved using proximal gradient-based algorithms \cite{condat_primal_2013}.

\blue{Other studies on signed graphs focus on the inference of \textit{balanced} signed graphs \cite{Dinesh2025, yokota2025efficient, yang_signed_2022} from statistical evidence.
However, these methods assume fully observed signals and are not designed for joint restoration and graph learning from degraded observations.}

In \cite{dittrich_learning_2020}, the sign matrix $\mathbf{S}$ is estimated by comparing the \blue{sum/difference} of signal vectors (see \eqref{eq:signed_distance}), which is highly sensitive to noise.
In the proposed method, however, we learn the signed graph from the complementary kernel matrix $\mathbf{K} = \mathbf{X}^\top\mathbf{L}_s\mathbf{X}$, which captures the cross-modal variation of the signal $\mathbf{X}$ on the spatial graph $\mathcal{G}_s$.
Therefore, instead of computing the signs from distance, we estimate the sign structure from the dominant eigenspace of $\mathbf{K}$ via subspace iteration, yielding a continuous relaxation that is compatible with end-to-end training.

\section{Twofold Signed Graph Learning and Signal Restoration for Multimodal Signals}
\label{sec:proposed}
In this section, we present our proposed method for jointly restoring the multimodal signals and learning its underlying twofold graph.
The method proceeds via alternating minimization over three coupled variables: The signal $\mathbf{X}$, the spatial graph Laplacian $\mathbf{L}_s$, and the modality graph Laplacian $\mathbf{L}_m$.
We first establish the probabilistic problem formulation (Section~\ref{subsec:formulation}), then describe each update as a standalone convex sub-problem (Sections~\ref{subsec:signal_update}--\ref{subsec:graph_update}), and conclude by showing how these iterative solvers are unrolled into a deep network (Section~\ref{subsec:unrolling}).
\blue{In contrast to prior works that either assume the twofold graph is known or restrict edge weights to be non-negative, our method learns a spatial unsigned graph and a modality signed graph simultaneously, capturing both positive and negative inter-modal correlations.
In addition, the proposed method is based on inversion-free iterative solvers, which are more efficient than the closed-form graph low-pass filters used in prior works \cite{kojima2026unrolling}.}

\subsection{Problem Formulation}
\label{subsec:formulation}
\blue{We first state the observation model and then derive the joint objective from a probabilistic signal prior.}

\subsubsection{Observation Model}
Analogous to the multimodal graph signal model in related works \cite{nagahama2022multimodal,kojima2026unrolling}, we consider the observation model
\begin{equation}
    \mathbf{Y} = \mathcal{A}(\mathbf{X}) + \mathbf{N},
\end{equation}
where $\mathbf{X}\in\mathbb{R}^{N\times M}$ is the unknown true multimodal graph signal: its $M$ columns are signals on the spatial graph $\mathcal{G}_s$, and its $N$ rows are signals on the modality graph $\mathcal{G}_m$.
The operator $\mathcal{A}(\cdot)$ models a degradation process (\eg, missing-data masking), and $\mathbf{N}$ is noise.
We aim to jointly recover $\mathbf{X}$ and the twofold graph $(\mathcal{G}_s, \mathcal{G}_m)$ from $\mathbf{Y}$.

\subsubsection{MAP Estimation via Matrix Normal Distribution}
In contrast to prior works on TGA \cite{nagahama2022multimodal,kojima2026unrolling} that lack explicit signal models, we place a \textit{matrix normal prior} \cite{gupta1999Matrix} on $\mathbf{X}$:
\begin{equation}
\mathbf{X} \sim \mathcal{MN}(\mathbf{0},\mathbf{L}_s^\dagger,\mathbf{L}_m^\dagger),
\end{equation}
where $\mathbf{L}_s \in \mathbb{R}^{N\times N}$ and $\mathbf{L}_m \in \mathbb{R}^{M\times M}$ are the row-wise and column-wise precision matrices, respectively, identified with the spatial and modality graph Laplacians.
Without loss of generality, we assume the mean of $\mathbf{X}$ is zero.
The probability density of $\mathbf{X}$ is then
\begin{equation}
\label{eq:pdf_mn}
    p(\mathbf{X}) \propto \exp\!\left(-\tfrac{1}{2}\text{tr}\!\left[\mathbf{L}_m\mathbf{X}^{\top}\mathbf{L}_s\mathbf{X}\right]\right).
\end{equation}
The term $\text{tr}(\mathbf{L}_m\mathbf{X}^\top\mathbf{L}_s\mathbf{X})$ simultaneously penalizes spatial variation on $\mathcal{G}_s$ and cross-modal variation on $\mathcal{G}_m$.
\blue{Intuitively, two entries of $\mathbf{X}$ are assumed to co-vary strongly when their nodes are connected in $\mathcal{G}_s$ and their modalities are connected in $\mathcal{G}_m$. A signal that is smooth on both graphs thus has a high prior probability.}
Assuming i.i.d.\ noise $N_{ij}\sim\mathcal{N}(0,\sigma^2)$, the likelihood is
\begin{equation}
\label{eq:pdf_YX}
    p(\mathbf{Y}|\mathbf{X})\propto \exp\!\left(-\tfrac{1}{2\sigma^2}\|\mathbf{Y}-\mathcal{A}(\mathbf{X})\|_F^{2}\right).
\end{equation}
Applying Bayes' rule, the MAP estimator minimizes
\begin{align}
\label{eq:MAP_estimate}
\mathbf{X}^\star = \arg\min_\mathbf{X}\!\left(-\log p(\mathbf{Y}|\mathbf{X})-\log p(\mathbf{X})\right).
\end{align}

\subsubsection{Joint Optimization Objective}
Substituting \eqref{eq:pdf_mn} and \eqref{eq:pdf_YX} into \eqref{eq:MAP_estimate}, the MAP objective for $\mathbf{X}$ is
\begin{equation}
J(\mathbf{X}) = \tfrac{1}{2\sigma^2}\|\mathbf{Y}-\mathcal{A}(\mathbf{X})\|_F^2+\tfrac{1}{2}\text{tr}(\mathbf{L}_m\mathbf{X}^\top\mathbf{L}_s\mathbf{X}).
\end{equation}
For the missing-data case we set $\mathcal{A}(\mathbf{X})=\mathbf{M}\circ\mathbf{X}$, where $\mathbf{M}\in\{0,1\}^{N\times M}$ is a binary observation mask.
Since $\mathbf{L}_s$ and $\mathbf{L}_m$ are typically unknown, we jointly minimize over all three variables:
\begin{equation}
    \label{eq:original_objective}
    \min_{\mathbf{X},\,\mathbf{L}_s,\,\mathbf{L}_m \in \mathcal{L}}\; \frac{1}{2}\|\mathbf{M}\circ(\mathbf{Y}-\mathbf{X})\|_F^2 + \frac{\mu}{2}\text{tr}(\mathbf{L}_m\mathbf{X}^\top\mathbf{L}_s\mathbf{X}) + R_s(\mathbf{L}_s) + R_m(\mathbf{L}_m),
\end{equation}
where $\mu$ balances data fidelity against the twofold-graph regularization term, and $R_s(\cdot)$, $R_m(\cdot)$ are regularizers on the spatial and modality graphs, \eg, Frobenius-norm penalties promoting sparsity and controlled edge magnitudes.
\blue{In contrast to the additive regularization in \eqref{eq:additive_tga}, the trace term couples the two graphs multiplicatively: The modality graph determines on which modality pairs the spatial smoothness is enforced, and vice versa.}

\subsection{Signal Reconstruction via Unrolled Conjugate Gradient}
\label{subsec:signal_update}
\blue{This subsection derives the sub-problem for the signal $\mathbf{X}$ and presents its solver.}

\subsubsection{Formulation}
Due to the coupling of $\mathbf{X}$, $\mathbf{L}_s$, and $\mathbf{L}_m$ in the trace term, the objective \eqref{eq:original_objective} is non-convex in general.
We therefore adopt an alternating minimization strategy, updating one variable at a time while holding the others fixed.

With $\mathbf{L}_s$ and $\mathbf{L}_m$ fixed, \eqref{eq:original_objective} reduces to the convex quadratic sub-problem
\begin{equation}
\label{eq:signal_restore}
\min_\mathbf{X}\ F(\mathbf{X})=\frac{1}{2}\|\mathbf{M}\circ(\mathbf{Y}-\mathbf{X})\|_F^2+\frac{\mu}{2}\text{tr}(\mathbf{L}_m\mathbf{X}^\top\mathbf{L}_s\mathbf{X}).
\end{equation}
Setting $\nabla_\mathbf{X}F=\mathbf{0}$ yields the optimality condition
\begin{align}
\label{eq:sylvester}
    \mu\mathbf{L}_s\mathbf{X}\mathbf{L}_m + \mathbf{M}\circ\mathbf{X} = \mathbf{M}\circ\mathbf{Y}.
\end{align}
\blue{The solution acts as a low-pass filter on both graphs at once: Observed entries are anchored by the mask, while missing entries are filled with information propagated through $\mathcal{G}_s$ and $\mathcal{G}_m$.}
If $\mathbf{M}=\mathbf{1}_N\mathbf{1}_M^\top$, \eqref{eq:sylvester} is a standard Sylvester equation $\mathbf{AXB}+\mathbf{X}=\mathbf{C}$, solvable in closed form \cite{zhou2012Preconditioned}.
For an arbitrary binary mask, the vectorized form of \eqref{eq:sylvester} is a large-scale linear system of size $NM\times NM$ with Kronecker structure $(\mu\mathbf{L}_m\otimes\mathbf{L}_s + \operatorname{diag}(\operatorname{vec}(\mathbf{M})))\operatorname{vec}(\mathbf{X}) = \operatorname{vec}(\mathbf{M}\circ\mathbf{Y})$.
However, directly solving this system is computationally prohibitive for large $N$ and $M$ due to the $O((NM)^3)$ complexity of naive linear solvers and the $O((NM)^2)$ memory requirement of storing the full system matrix.

Since $\mathbf{L}_s,\mathbf{L}_m\succeq\mathbf{0}$ and $M_{ij}\in\{0,1\}$, the operator $\mathcal{H}(\mathbf{X}) \triangleq \mathbf{M}\circ\mathbf{X}+\mu\mathbf{L}_s\mathbf{X}\mathbf{L}_m$
in \eqref{eq:sylvester} is self-adjoint and positive semidefinite on $\mathbb{R}^{N\times M}$; 
we solve \eqref{eq:sylvester} with the conjugate gradient (CG) method, which is efficient for large-scale linear systems with such properties.

\subsubsection{Conjugate Gradient Solver}
We solve \eqref{eq:sylvester} via the conjugate gradient (CG) algorithm (Algorithm~\ref{alg:sylvester}).
Starting from an initial estimate $\mathbf{X}^{(0)}$ with residual $\mathbf{R}^{(0)}=\mathbf{M}\circ\mathbf{Y}-\mathcal{H}(\mathbf{X}^{(0)})$ and search direction $\mathbf{P}^{(0)}=\mathbf{R}^{(0)}$, the algorithm iterates while updating the step size $\kappa$ and conjugate direction coefficient $\xi$.

CG is well-suited to this sub-problem for two reasons.
First, it exploits the Kronecker structure of the system: The linear operator is applied as matrix products $\mathbf{L}_s\mathbf{X}\mathbf{L}_m$, without forming the full $NM\times NM$ system explicitly.
Second, every operation in Algorithm~\ref{alg:sylvester}---inner products, matrix multiplications, and scalar rescalings---is strictly differentiable, enabling gradient backpropagation through the truncated CG iterations when the solver is later embedded in the unrolled network.

\begin{algorithm}[t]
  \caption{Conjugate Gradient for signal reconstruction \eqref{eq:sylvester}}
  \label{alg:sylvester}
  \begin{algorithmic}
  \REQUIRE $\mathbf{X}^{(0)}$, $\mathbf{R}^{(0)}=\mathbf{M}\circ\mathbf{Y}-\mathcal{H}(\mathbf{X}^{(0)})$, $\mathbf{P}^{(0)}=\mathbf{R}^{(0)}$
  \ENSURE $\mathbf{X}^{(k+1)}$
  \STATE $k\leftarrow0$
  \WHILE{ $\|\mathbf{R}^{(k)}\|_F \ge\epsilon$ }
  \STATE $\mathbf{S}^{(k)}\leftarrow\mathcal{H}(\mathbf{P}^{(k)})$
  \STATE $\kappa^{(k)} \leftarrow \|\mathbf{R}^{(k)}\|_F^2\ /\ \langle\mathbf{P}^{(k)},\mathbf{S}^{(k)}\rangle_F$
  \STATE $\mathbf{X}^{(k+1)} \leftarrow \mathbf{X}^{(k)} + \kappa^{(k)}\mathbf{P}^{(k)}$
  \STATE $\mathbf{R}^{(k+1)} \leftarrow \mathbf{R}^{(k)} - \kappa^{(k)}\mathbf{S}^{(k)}$
  \STATE $\xi^{(k)} \leftarrow \|\mathbf{R}^{(k+1)}\|_F^2\ /\ \|\mathbf{R}^{(k)}\|_F^2$
  \STATE $\mathbf{P}^{(k+1)} \leftarrow \mathbf{R}^{(k+1)} + \xi^{(k)} \mathbf{P}^{(k)}$
  \STATE $k \leftarrow k+1$
  \ENDWHILE
  \end{algorithmic}
\end{algorithm}

\subsection{Twofold Graph Learning}
\label{subsec:graph_update}
\blue{We next derive the graph updates: We first describe the common unsigned formulation, then extend it to signed modality graphs, and present the PDHG solver.}

\subsubsection{Formulation}
With $\mathbf{X}$ and one Laplacian fixed (\eg, $\mathbf{L}_m$ when updating the spatial graph), the sub-problem for the target Laplacian $\mathbf{L}$ is
\begin{equation}
\label{eq:GL_objective}
\min_{\mathbf{L}\in\mathcal{L}} \theta_1\text{tr}({\mathbf{L}}\mathbf{K}) + \theta_2\|\mathbf{L}\|_F^2 - \theta_3\mathbf{1}^\top\log(\text{diag}(\mathbf{L})),
\end{equation}
where $\theta_1$ scales the kernel fitting term, $\theta_2$ controls edge-weight magnitude, and $\theta_3>0$ prevents degenerate zero-degree solutions.
\blue{Intuitively, the first term encourages larger edge weights between nodes with similar variation in $\mathbf{K}$. 
The second term controls the overall magnitude of the edge weights. 
The log-barrier avoids isolated nodes.}
This has the same form as the single-graph learning objective in \cite{dong_learning_2016}; the coupling between the spatial and modality topologies is encoded entirely in the kernel $\mathbf{K}$.

\subsubsection{The Role of the Target Kernel ($\mathbf{K}$)}
The kernel $\mathbf{K}$ is the key to coupling the two graph sub-problems.
When learning the spatial graph, $\mathbf{K}=\mathbf{X}\mathbf{L}_m\mathbf{X}^\top\in\mathbb{R}^{N\times N}$; when learning the modality graph, $\mathbf{K}=\mathbf{X}^\top\mathbf{L}_s\mathbf{X}\in\mathbb{R}^{M\times M}$.
In either case, the $(i,j)$ entry of $\mathbf{K}$,
\begin{equation}
K_{ij} = \mathbf{x}_i^\top\mathbf{L}_{c}\mathbf{x}_j = \sum_{p<q}W^{(c)}_{pq}(X_{ip}-X_{iq})(X_{jp}-X_{jq}),\quad \forall c\in\{s,m\}
\end{equation}
measures the co-variation of signals at nodes $i$ and $j$ as filtered through the complementary graph $\mathbf{L}_c$, mathematically coupling the two topologies.
Expanding the trace term in \eqref{eq:GL_objective} yields,
\begin{align}
    \text{tr}(\mathbf{L}\mathbf{K}) &= \frac{1}{2}\sum_{ij}W_{ij}(K_{ii}+K_{jj}-2K_{ij}) \nonumber \\
    &= \frac{1}{4}\sum_{ij}\sum_{pq}W_{ij}W'_{pq}\bigl[(X_{ip}-X_{jp})-(X_{iq}-X_{jq})\bigr]^2.
\end{align}
Thus, minimizing $\text{tr}(\mathbf{LK})$ enforces a structural regularization on the target weights $W_{ij}$, such that nodes with similar variation on the complementary graphs are encouraged to have large weights between them.

\subsubsection{Continuous Relaxation for Signed Modality Graphs}
\label{subsubsec:signed}
The feasible set $\mathcal{L}$ in \eqref{eq:GL_objective} enforces $W_{ij}\ge0$, which is appropriate for the spatial graph.
For the modality graph, however, unsigned edges cannot represent negative inter-modal correlations.
We therefore extend the formulation to \textit{signed} modality graphs.

We parameterize the signed modality Laplacian as $\bar{\mathbf{L}} = \mathbf{S}\circ\mathbf{L}_m$, where \blue{$\mathbf{S}\in\{-1,+1\}^{M\times M}$ is a symmetric polarity matrix with $S_{ii}=1$ (\eg, $\mathbf{S}=\mathbf{s}\mathbf{s}^\top$ for a two-group polarity vector $\mathbf{s}\in\{-1,+1\}^{M}$)} and $\mathbf{L}_m=\text{diag}(\mathbf{W}_m\mathbf{1})-\mathbf{W}_m$ is a standard non-negative Laplacian encoding the edge magnitudes.
The following proposition shows that learning $\bar{\mathbf{L}}$ reduces to learning an unsigned Laplacian on a sign-modulated kernel, making the problem fully tractable.

\begin{prop}[Trace Equivalence under Hadamard Product]
\label{prop:trace_equiv}
Let $\mathbf{K}\in\mathbb{R}^{M\times M}$ be symmetric and $\mathbf{S}\in\{-1,1\}^{M\times M}$ be a symmetric sign matrix with $S_{ii}=1$ for all $i$. For any positive graph Laplacian $\mathbf{L}_m=\operatorname{diag}(\mathbf{W}_m\mathbf{1})-\mathbf{W}_m$ with $\mathbf{W}_m\ge0$, and $\mathbf{K}_{\mathrm{mag}}=\mathbf{S}\circ\mathbf{K}$:
\begin{equation}
    \label{eq:trace_equiv}
    \operatorname{tr}(\bar{\mathbf{L}}\,\mathbf{K}) = \operatorname{tr}(\mathbf{L}_m\,\mathbf{K}_{\mathrm{mag}}).
\end{equation}
\end{prop}
\begin{proof}
Expanding $\bar{\mathbf{L}}=\operatorname{diag}(\mathbf{W}_m\mathbf{1})-\mathbf{S}\circ\mathbf{W}_m$:
\begin{equation}
    \operatorname{tr}(\bar{\mathbf{L}}\,\mathbf{K}) = \operatorname{tr}(\operatorname{diag}(\mathbf{W}_m\mathbf{1})\,\mathbf{K}) - \operatorname{tr}((\mathbf{S}\circ\mathbf{W}_m)\mathbf{K}).
\end{equation}
For the adjacency term, since $\mathbf{K}$ and $\mathbf{S}$ are symmetric,
\begin{equation}
    \operatorname{tr}((\mathbf{S}\circ\mathbf{W}_m)\mathbf{K}) = \sum_{i,j}W_{m,ij}S_{ij}K_{ij} = \operatorname{tr}(\mathbf{W}_m\,\mathbf{K}_{\mathrm{mag}}).
\end{equation}
For the degree term, $S_{ii}=1$ implies $(\mathbf{K}_{\mathrm{mag}})_{ii}=K_{ii}$, so
\begin{equation}
    \operatorname{tr}(\operatorname{diag}(\mathbf{W}_m\mathbf{1})\,\mathbf{K}) = \operatorname{tr}(\operatorname{diag}(\mathbf{W}_m\mathbf{1})\,\mathbf{K}_{\mathrm{mag}}).
\end{equation}
Recombining: $\operatorname{tr}(\bar{\mathbf{L}}\,\mathbf{K}) = \operatorname{tr}((\operatorname{diag}(\mathbf{W}_m\mathbf{1})-\mathbf{W}_m)\,\mathbf{K}_{\mathrm{mag}}) = \operatorname{tr}(\mathbf{L}_m\,\mathbf{K}_{\mathrm{mag}})$.
\end{proof}

Proposition~\ref{prop:trace_equiv} reduces the signed graph learning problem to the standard unsigned objective \eqref{eq:GL_objective} applied to the modulated kernel $\mathbf{K}_{\mathrm{mag}}=\mathbf{S}\circ\mathbf{K}$, provided that the sign matrix $\mathbf{S}$ is estimated.

Obtaining the exact bipartite partition $\mathbf{S}$ from data is in general NP-hard\blue{. In the literature, the sign matrix is often estimated from pairwise signal distances \cite{dittrich_learning_2020}, which is sensitive to noise and local discrepancies.}
\blue{In our method, we instead estimate the sign structure from the dominant eigenspace of $\mathbf{K}$, which captures the global polarity structure of the modalities.}
\blue{Specifically}, we estimate a continuous approximation via subspace iteration.
Let $\mathbf{P}\in\mathbb{R}^{M\times r}$ collect the $r$ leading eigenvectors of $\mathbf{K}$, and let $\tilde{\mathbf{P}}$ be obtained by row-normalizing $\mathbf{P}$.
We form the rank-$r$ approximation of $\mathbf{S}$ as the cosine similarity matrix of $\tilde{\mathbf{P}}$:
\begin{equation}
    \label{eq:spectral_sign}
    \mathbf{C} = \tilde{\mathbf{P}}\tilde{\mathbf{P}}^\top \in [-1,1]^{M\times M},
\end{equation}
which serves as a continuous surrogate for $\mathbf{S}$: The entry $C_{ij}=\cos(\theta_{ij})$ is the cosine similarity between nodes $i$ and $j$ in the dominant spectral subspace of $\mathbf{K}$, capturing the globally consistent polarity structure while suppressing local noise.
\blue{The idea is analogous to spectral clustering \cite{vonluxburg2007tutorial}: For a nearly balanced signed structure, the leading eigenvectors of $\mathbf{K}$ separate the modalities into anti-correlated groups. 
This global perspective makes our approach more robust to noisy and incomplete data.}
Replacing $\mathbf{S}$ with $\mathbf{C}$ in \eqref{eq:trace_equiv}, the modulated kernel becomes $\mathbf{K}_{\mathrm{mag}} = \mathbf{C}\circ\mathbf{K}$.

\subsubsection{Primal-Dual Hybrid Gradient Solver}
\label{subsubsec:pdhg}
We solve \eqref{eq:GL_objective}, where $\mathbf{K}$ replaced by $\mathbf{K}_{\mathrm{mag}}$ for the modality graph, by reparametrizing the graph Laplacian matrix with the adjacency matrix $\mathbf{L}=\text{diag}(\mathbf{W}\mathbf{1})-\mathbf{W}$:
\begin{align}
\label{eq:W_opt}
    \min_{\mathbf{W}\in\mathcal{W}} \theta_1(\mathbf{d}^\top\text{diag}(\mathbf{K})-\text{tr}(\mathbf{W}\mathbf{K})) + \theta_2(\|\mathbf{d}\|_2^2+\|\mathbf{W}\|_F^2) - \theta_3\mathbf{1}^\top\log(\mathbf{d})\nonumber\\
    \text{s.t.}\;  \mathbf{d} = \mathbf{W}\mathbf{1}.
\end{align}
Introducing a dual variable $\mathbf{u}$ to decouple the degree constraint $\mathbf{d}=\mathbf{W}\mathbf{1}$, the problem becomes finding the saddle point of
\begin{equation}
\label{eq:saddle_point_problem}
    \mathcal{F}(\mathbf{W},\mathbf{d},\mathbf{u}) = f(\mathbf{W}) + h(\mathbf{d}) + \mathbf{u}^\top(\mathbf{W}\mathbf{1}-\mathbf{d}),
\end{equation}
where \blue{$f(\mathbf{W})=-\theta_1\text{tr}(\mathbf{W}\mathbf{K})+\theta_2\|\mathbf{W}\|_F^2$} is the smooth primal objective, \blue{$h(\mathbf{d})=\theta_1\mathbf{d}^\top\text{diag}(\mathbf{K})+\theta_2\|\mathbf{d}\|^2_2-\theta_3\mathbf{1}^\top\log(\mathbf{d})$} is the degree-specific regularization.
\blue{By splitting $\mathbf{d}$ as a separate variable, the log-barrier reduces to an element-wise problem with a closed-form solution. 
The dual variable $\mathbf{u}$ ensures $\mathbf{d}$ is consistent with $\mathbf{W}\mathbf{1}$.}
We solve \eqref{eq:saddle_point_problem} via an accelerated primal-dual hybrid gradient (PDHG) method (Algorithm~\ref{alg:GL_pdhg}) as follows.

\begin{enumerate}
    \item \textit{Primal Descent (Adjacency Matrix Update):}
        The adjacency matrix is updated with a gradient step on $f$ followed by projection onto $\mathcal{W}$:
        \begin{equation}
        \label{eq:proj_op}
            \mathbf{W}^{(k+1)}=\text{Proj}_\mathcal{W}\!\left(\mathbf{W}^{(k)}-\tau\!\left(\nabla f(\mathbf{W}^{(k)})+\mathcal{S}^\ast(\mathbf{u}^{(k)})\right)\right),
        \end{equation}
        where $\tau$ is the primal step-size, $\mathcal{S}^\ast(\mathbf{u})=\mathbf{u}\mathbf{1}^\top+\mathbf{1}\mathbf{u}^\top$ is the adjoint of the degree operator $\mathcal{S}(\mathbf{W})=\mathbf{W}\mathbf{1}=\mathbf{d}$, and $\text{Proj}_\mathcal{W}$ enforces non-negativity, symmetry, and a zero diagonal.
        Since \blue{$f(\mathbf{W})=\theta_2\|\mathbf{W}\|_F^2-\theta_1\text{tr}(\mathbf{W}\mathbf{K})$}, its gradient is
        \begin{equation}
        \label{eq:grad_f}
            \blue{\nabla f(\mathbf{W}) = 2\theta_2\mathbf{W} - \theta_1\mathbf{K}.}
        \end{equation}
    \item \noindent\textit{Extrapolation:}
        To achieve an $\mathcal{O}(1/k)$ convergence rate, we extrapolate the updated adjacency matrix:
        \begin{align}
            \bar{\mathbf{W}} = 2\mathbf{W}^{(k+1)} - \mathbf{W}^{(k)}, \qquad \bar{\mathbf{d}}^{(k+1)} = \bar{\mathbf{W}}\mathbf{1}.
        \end{align}
    \item \noindent\textit{Dual Ascent (Degree Update):}
        Setting $\mathbf{v}=(\mathbf{u}^{(k)}+\sigma\bar{\mathbf{d}}^{(k+1)})/\sigma$ and applying Moreau's identity, the dual update reduces to the element-wise separable proximal problem
        \begin{equation}
        \label{eq:dual_ascent_obj}
            \blue{\min_{d_i}\!\left(\theta_2 d_i^2 + \theta_1 K_{ii}d_i-\theta_3\log(d_i)+\frac{1}{2\sigma}(d_i-v_i)^2\right),}
        \end{equation}
        where $v_i = u_i/\sigma + \bar{d}_i$. This is a one-dimensional strongly convex problem with the closed-form solution
        \begin{equation}
            \blue{d_i^{(k+1)} = \frac{v_i-\sigma\theta_1 K_{ii}+\sqrt{(\sigma\theta_1 K_{ii}-v_i)^2+4\sigma\theta_3(1+2\sigma\theta_2)}}{2(1+2\sigma\theta_2)},}
        \end{equation}
        after which the dual variable is updated as:
        \begin{equation}
            \mathbf{u}^{(k+1)}=\mathbf{v}-\sigma\mathbf{d}^{(k+1)}.
        \end{equation}
\end{enumerate}

\begin{algorithm}[t]
    \caption{Primal-Dual Hybrid Gradient for Graph Learning}
    \label{alg:GL_pdhg}
    \begin{algorithmic}
    \REQUIRE $\mathbf{K}$ (or $\mathbf{K}_{\mathrm{mag}}$ for the modality graph), $\mathbf{W}^{(0)}$, $\mathbf{u}^{(0)}$
    \ENSURE $\mathbf{L}$
    \WHILE{ $\|\mathbf{W}^{(k+1)} - \mathbf{W}^{(k)}\|_F/\|\mathbf{W}^{(k)}\|_F\blue{>}\epsilon_{g}$ }
    \STATE $\mathbf{W}^{(k+1)} \leftarrow \text{Proj}_{\mathcal{W}}\!\left( \mathbf{W}^{(k)} - \tau \left( \blue{2\theta_2\mathbf{W}^{(k)} - \theta_1\mathbf{K}} + \mathcal{S}^*(\mathbf{u}^{(k)}) \right)\right)$
    \STATE $\bar{\mathbf{W}} \leftarrow 2 \mathbf{W}^{(k+1)} - \mathbf{W}^{(k)}$
    \STATE $\mathbf{v} \leftarrow (\mathbf{u}^{(k)}+\sigma\bar{\mathbf{W}}\mathbf{1})/\sigma$
    \STATE $d_{i}^{(k+1)} \leftarrow\blue{\frac{v_i-\sigma\theta_1 K_{ii}+\sqrt{(\sigma\theta_1 K_{ii}-v_i)^2+4\sigma\theta_3(1+2\sigma\theta_2)}}{2(1+2\sigma\theta_2)}}$
    \STATE $\mathbf{u}^{(k+1)} \leftarrow (\mathbf{u}^{(k)}+\sigma\bar{\mathbf{W}}\mathbf{1}) - \sigma\mathbf{d}^{(k+1)}$
    \ENDWHILE
    \STATE $\mathbf{d}_{\text{final}} \leftarrow \mathbf{W}^{(k+1)} \mathbf{1}$
    \STATE $\mathbf{L} \leftarrow \text{diag}(\mathbf{d}_\text{final})-\mathbf{W}^{(k+1)}$
    \end{algorithmic}
\end{algorithm}

\subsection{Computational Complexity}
\label{subsec:complexity}

We summarize the per-layer cost of each module in terms of $N$ (spatial nodes) and $M$ (modalities).

\begin{enumerate}
    \item \textit{Signal Reconstruction:}
    Each iteration applies the operator $\mathbf{L}_s\mathbf{P}^{(k)}\mathbf{L}_m$ via two successive matrix products, costing $O(N^2M+NM^2)$ per iteration and $O(N^2M+NM^2)$ total for a fixed number of steps.

    \item \textit{Kernel and Polarity Estimation:}
    Computing the kernel \blue{($\mathbf{K}=\mathbf{X}^\top\mathbf{L}_s\mathbf{X}\in\mathbb{R}^{M\times M}$ or $\mathbf{K}=\mathbf{X}\mathbf{L}_m\mathbf{X}^\top\in\mathbb{R}^{N\times N}$) costs} $O(NM(N+M))$, and the polarity surrogate $\mathbf{C}=\tilde{\mathbf{P}}\tilde{\mathbf{P}}^\top$ via subspace iteration \blue{costs} $O(M^2)$.

    \item \textit{Graph Learning:}
    Each iteration in the proposed PDHG algorithm updates \blue{an $N\times N$ (spatial) or $M\times M$ (modality) adjacency matrix in $O(N^2)$ or $O(M^2)$, respectively}.
\end{enumerate}
The \blue{overall} computational complexity is dominated by the kernel computation and the CG operator, giving $O(N^2M+NM^2)$.

\subsection{Deep Algorithm Unrolling}
\label{subsec:unrolling}
\blue{Finally, we assemble the solvers described above into a trainable feedforward network.}

\subsubsection{Network Architecture}
The alternating minimization over $(\mathbf{X},\mathbf{L}_s,\mathbf{L}_m)$ directly defines the forward pass of the unrolled network.
Each outer layer $l\in\{1,\dots,L\}$ comprises one full alternating cycle with three sub-modules executed in sequence:
\begin{enumerate}
    \item[(i)] \textbf{Modality Graph Module.} Compute $\mathbf{K}=\mathbf{X}^{(l-1)\top}\mathbf{L}_s^{(l-1)}\mathbf{X}^{(l-1)}$, form $\mathbf{K}_{\mathrm{mag}}=\mathbf{C}\circ\mathbf{K}$, and run $T$ PDHG inner iterations (Algorithm~\ref{alg:GL_pdhg}) with $\mathbf{K}_{\mathrm{mag}}$ to produce $\mathbf{L}_m^{(l)}$.
    \item[(ii)] \textbf{Spatial Graph Module.} Run $T$ PDHG inner iterations with unsigned kernel $\mathbf{K}=\mathbf{X}^{(l-1)}\mathbf{L}_m^{(l)}\mathbf{X}^{(l-1)\top}$ to produce $\mathbf{L}_s^{(l)}$.
    \item[(iii)] \textbf{Signal Reconstruction Module.} Run $K$ CG inner iterations (Algorithm~\ref{alg:sylvester}) with the updated pair $(\mathbf{L}_s^{(l)},\mathbf{L}_m^{(l)})$ to produce $\mathbf{X}^{(l)}$.
\end{enumerate}
The output of the final outer layer, $\mathbf{X}^{(L)}$, is the reconstructed signal.
We show the overview of the unrolled network architecture in Fig.~\ref{fig:network_architecture}.

\begin{figure}[t]
    \centering
    \includegraphics[width=0.95\textwidth]{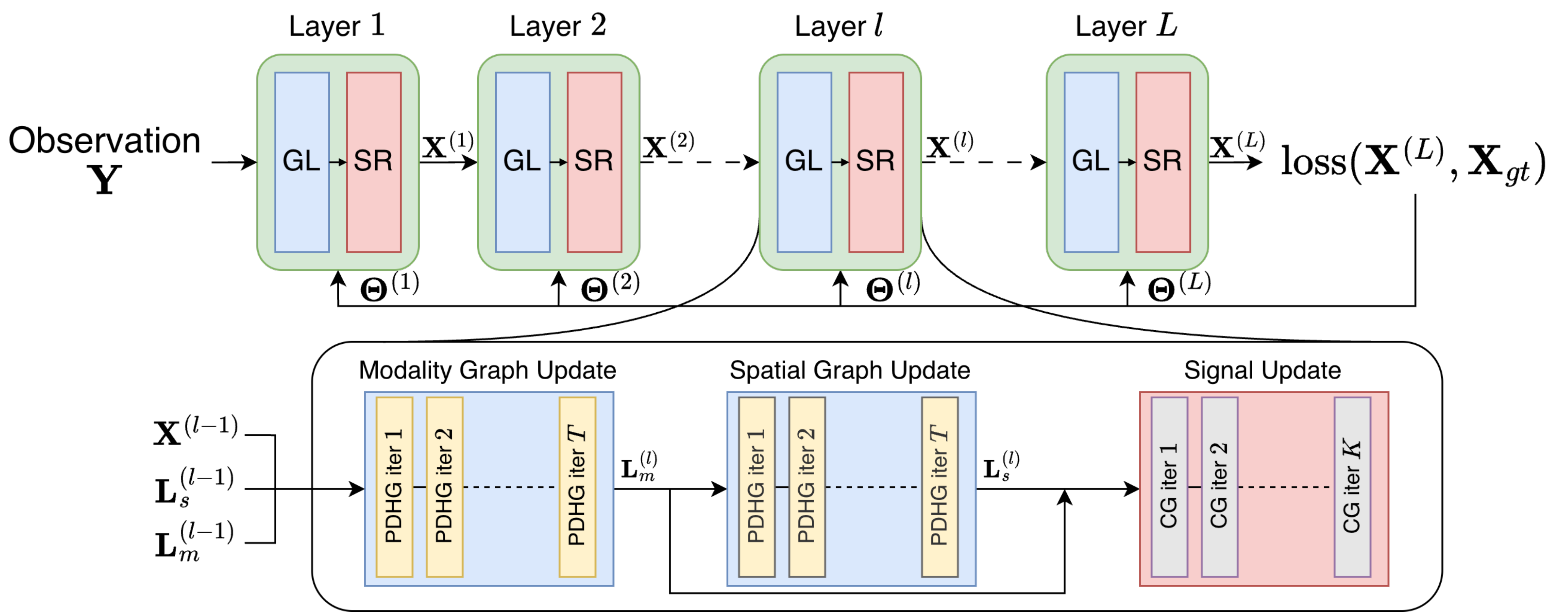}
    \caption{Overview of the unrolled network architecture. Each outer layer consists of a modality graph module, a spatial graph module, and a signal reconstruction module executed in sequence.}
    \label{fig:network_architecture}
\end{figure}

\subsubsection{Learnable Parameters}
By fixing the number of inner iterations and promoting all algorithmic hyperparameters to layer-wise trainable weights, we obtain the following parameter sets.
For the two graph learning modules at outer layer $l$:
\begin{equation}
    \mathbf{\Theta}_{\mathrm{GL}}^{(l)} := \left\{ \bigl\{ \theta_1^{(t,l)}, \theta_2^{(t,l)}, \theta_3^{(t,l)}, \tau^{(t,l)}, \sigma^{(t,l)} \bigr\}_c \;\middle|\; c \in \{s, m\}, t = 1, \dots, T \right\},
\end{equation}
Here, $\theta_1$ is a learnable scale for the kernel fitting term, and $\theta_2, \theta_3$ are the regularization weights in \eqref{eq:GL_objective}. $\tau, \sigma$ are the PDHG primal and dual step-sizes in Algorithm~\ref{alg:GL_pdhg}. 
Each parameter is learned separately for the spatial and modality modules.
For the signal reconstruction module at outer layer $l$:
\begin{equation}
    \mathbf{\Theta}_{\mathrm{SR}}^{(l)} := \bigl\{\kappa^{(k,l)},\,\xi^{(k,l)},\,\mu^{(l)}\bigr\}_{k=1}^{K},
\end{equation}
where $\kappa^{(k)},\xi^{(k)}$ are the per-iteration CG step size and conjugate direction coefficient in Algorithm~\ref{alg:sylvester}, and $\mu^{(l)}$ is the twofold-graph regularization weight in \blue{\eqref{eq:signal_restore}}.
Here $T$ and $K$ are the numbers of inner PDHG and CG iterations, respectively.
The full trainable parameter set is $\mathbf{\Theta}=\{\mathbf{\Theta}_{\mathrm{GL}}^{(l)},\,\mathbf{\Theta}_{\mathrm{SR}}^{(l)}\}_{l=1}^{L}$.
\blue{In other words, the unrolled network is the iterative solver with trainable regularization parameters and step sizes that vary per layer.
Training learns when to tighten or loosen the regularization strength in a data-driven manner rather than a black-box mapping of inputs to outputs. 
This keeps the parameter count small, and every learned value remains interpretable.}

\subsubsection{End-to-End Training}
Given training pairs $\{(\mathbf{Y}^{(n)},\mathbf{X}_{\mathrm{gt}}^{(n)})\}_{n=1}^{N_{\mathrm{train}}}$, we minimize the mean squared error on the missing entries:
\begin{equation}
    \mathcal{J}(\mathbf{\Theta}) = \frac{1}{NM}\,\mathbb{E}_n\!\left[\|\mathbf{M}^c\circ(\mathbf{X}^{(L)}-\mathbf{X}_{\mathrm{gt}})\|_F^2\right],
\end{equation}
where $\mathbf{M}^c=\mathbf{1}\mathbf{1}^\top-\mathbf{M}$ selects only the unobserved entries.
Since every component of the forward pass---CG iterations, PDHG updates---is (sub-)differentiable, the network is trained end-to-end via backpropagation.
The overall number of trainable parameters is $O(L(T+K))$, which is independent of the signal dimensions $N$ and $M$, making the method scalable to large graphs.

\section{Experiments}
\label{sec:experiments}
In this section, we evaluate the proposed method on multimodal signal denoising and interpolation against competing baselines in two settings: (i) synthetic multimodal graph signals with varying noise levels and missing rates, and (ii) interpolation of \blue{two real-world datasets}.

We fix the number of inner PDHG iterations to $T=5$ and CG iterations to $K=5$ in all experiments.
The baselines include the following methods:
\begin{enumerate}
    \item Model based methods: graph low-pass filtering (GLF), singular value decomposition (SVD)
    \item Data-driven methods: autoencoder (AE) \cite{vincent2010stacked}, graph convolutional network (GCN) \cite{kipf2017semisupervised}
    \item Algorithm-unrolling methods: TGSR-DAU \cite{nagahama2022multimodal}, LLAP-DAU \cite{kojima2026unrolling}\blue{, NestDAU \cite{nagahama2022nested}, and joint inpainting and graph learning (IDGL) \cite{batreddy2025inpainting}}
\end{enumerate}
TGSR-DAU is a recent unrolling method for multimodal graph signal recovery, which requires a pair of pre-specified spatial graph and modality graph. LLAP-DAU is an unrolling method that learns twofold graphs from observations, but \blue{the learned graphs are restricted to be unsigned}.
\blue{The GCN baseline uses the renormalized adjacency with self-loops \cite{kipf2017semisupervised}; since the propagation can be performed either on the spatial graph or on the modality graph, we report both variants.
NestDAU restores each modality independently: Each column is treated as a single graph signal on a six-nearest-neighbor spatial graph constructed from the observations.
IDGL learns a single unsigned spatial graph Laplacian together with the signal, using a difference prior along the columns; since the columns of the synthetic datasets are unordered modalities, the difference operator is applied across all columns.}

All methods are implemented in PyTorch and trained on a NVIDIA RTX 5060 GPU.

\subsection{\blue{Synthetic Datasets}}
\blue{We first evaluate on synthetic data, where the ground-truth signals and graphs are known.}

We compare the performance of the proposed method against baselines with \blue{three} types of synthetic datasets:
\begin{enumerate}
    \item Matrix Normal dataset: Signals are sampled from a matrix normal distribution with synthetic precision matrices
    \item Piecewise Smooth dataset: Signals are generated with smooth spatial patterns and piecewise smooth modality variations.
    \item \blue{Signed Piecewise Smooth dataset: Signals are piecewise smooth over modality communities whose inter-community correlations are governed by a signed modality graph.}
\end{enumerate}
Specifically, \blue{each} synthetic dataset is generated in two stages.
First, we create a twofold graph. Then, we \blue{generate} the clean signal $\mathbf{X}$ \blue{based on the twofold graph}, and add i.i.d. Gaussian noise to obtain the noisy observation $\mathbf{Y}$.

The spatial graph $\mathcal{G}_s$ is a six-nearest neighbor graph of $N$ nodes randomly placed in a 2-D space $[0,1]\times [0,1]$. The modality graph $\mathcal{G}_m$ is a community graph with $M$ nodes randomly assigned to six communities. Each community in the modality graph corresponds to a single modality. 
Finally, an i.i.d. noise matrix $\mathbf{E}\in\mathbb{R}^{N\times M}$ is generated with $E_{ij}\sim\mathcal{N}(0,\sigma^2)$, where we control $\sigma^2$ to achieve four levels of signal to noise ratio (SNR) in dB: $[5,\ 10,\ 15,\ 20]$.
The observation mask matrix $\mathbf{M}\in\mathbb{R}^{N\times M}$ for interpolation is generated with five different missing rates in $\%$: $[10,\ 20,\ 30,\ 40,\ 50]$.
We visualize the sample signals from the \blue{three} datasets in Fig.~\ref{fig:synth_sample}. 

\blue{We describe how the modality graph and the corresponding clean signals are generated for each of the three datasets below.}

\subsubsection{Matrix Normal Dataset} 
We consider the stochastic block model \cite{holland_stochastic_1983} for the connectivity of the modality graph: Intra-community edges and inter-community edges are generated with probability $p=0.3$ and $q=0.5$, respectively. The magnitudes of edge weights are drawn uniformly from $[0.1,1.0]$. 
Then, the communities are split randomly into two groups, where edge weights between the communities within the same groups remain positive, and edge weights between communities across groups are multiplied by $-1$.
A sample matrix observation $\mathbf{X}\in\mathbb{R}^{N\times M}$ is drawn from the distribution $\mathcal{MN}(\mathbf{0},\mathbf{L}_s^\dagger,\mathbf{L}_m^\dagger)$.

\subsubsection{Piecewise Smooth Dataset} 
We sample six base spatial signals $\mathbf{x}_{j}\sim\mathcal{N}(\mathbf{0},\mathbf{L}_s^\dagger)$, one for each modality community $j\in\{1,\ldots,6\}$. 
The $m$th spatial signal for community $j$ is then generated as
\begin{equation}
    \mathbf{x}_{j,m} = \mathbf{x}_{j} + \phi_{j}(t_m)\mathbf{1},
\end{equation}
where $t_m \in [0,1]$ is the normalized modality node index. The additive signal $\phi_{j}(t_m)$ is a sinusoidal function introducing anti-phase oscillations between adjacent communities:
\begin{equation}
    \phi_{j}(t_m) = \begin{cases}
        \frac{1}{4}\sin(\frac{\pi}{2}t_m), & \text{if $j$ is odd}, \\
        \frac{1}{4}\sin(\frac{\pi}{2}t_m-\pi), & \text{if $j$ is even}.
    \end{cases}
\end{equation}

\blue{Note that the negative inter-modal correlations of the piecewise smooth dataset are uniform across community pairs rather than tied to specific signed edges.}

\subsubsection{Signed Piecewise Smooth Dataset}
\blue{To directly test the signed-graph contribution, we additionally introduce a dataset whose modality graph carries inter-community negative correlations that are explicitly tied to the signal generation process.}
\blue{Each pair of the six communities is joined with probability $0.5$, and each community receives a polarity $\sigma_j\in\{-1,+1\}$: The ground-truth edge weights of the modality graph are positive within communities and between joined communities with the same polarity, negative between joined communities with different polarities, and zero if they are disconnected.
Accordingly, the base signals $\tilde{\mathbf{x}}_{j}$, each marginally distributed as $\mathcal{N}(\mathbf{0},\mathbf{L}_s^\dagger)$, are drawn jointly so that their cross-community correlations match the ground-truth edge weights:
\begin{equation}
    \operatorname{corr}(\tilde{\mathbf{x}}_{j},\tilde{\mathbf{x}}_{j'}) =
    \begin{cases}
        \rho\,\sigma_j\sigma_{j'}, & \text{if communities $j$ and $j'$ are joined},\\
        0, & \text{otherwise},
    \end{cases}
\end{equation}
where $\rho=0.6$.
The additive offset $\phi_{j}(t_m)\mathbf{1}$ is replaced by a multiplicative modulation with polarity-dependent phase:
\begin{equation}
    \mathbf{x}_{j,m} = \bigl(1+\gamma\sin(\pi t_m + \varphi_j)\bigr)\,\tilde{\mathbf{x}}_{j},
\end{equation}
where $\gamma=0.25$, and $\varphi_j=0$ if $\sigma_j=+1$ and $\varphi_j=\pi$ otherwise.}

\begin{figure}
    \centering
    \subfigure[Matrix normal dataset]{\includegraphics[width=0.32\linewidth]{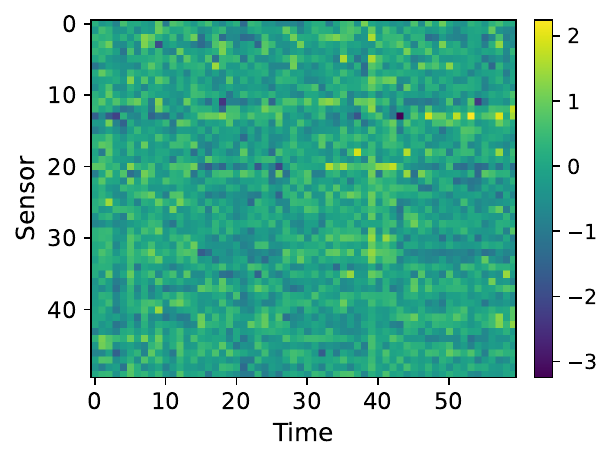}}
    \subfigure[Piecewise smooth dataset]{\includegraphics[width=0.32\linewidth]{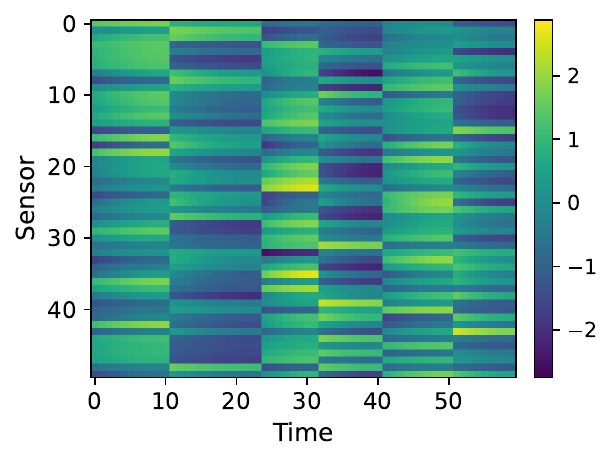}}
    \subfigure[\blue{Signed piecewise smooth dataset}]{\includegraphics[width=0.32\linewidth]{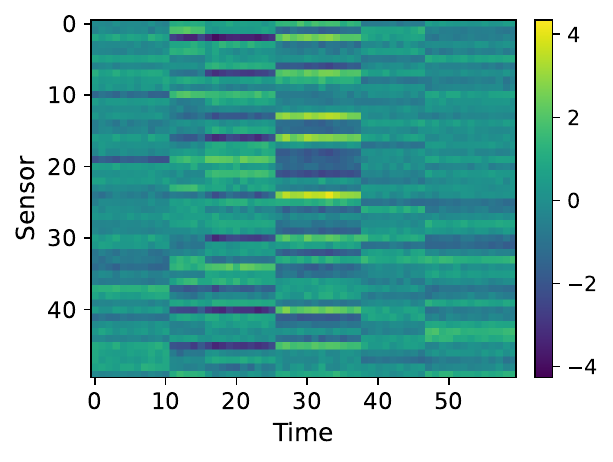}}
    \caption{Sample signals from the \blue{three} synthetic datasets.}
    \label{fig:synth_sample}
\end{figure}

\subsection{Experimental Setup and Results}

\subsubsection{Experimental Setup}
We generate five sets of twofold graphs with $N=50$ and $M=60$, and sample $T=20$ signals for each twofold graph to perform five-fold cross-validation. We set samples from one twofold graph as the validation set, and use the remaining samples as the training set at each cross-validation fold. 
For baselines requiring a known graph (TGSR-DAU), we constructed a six-nearest neighbor graph from the noisy observation $\mathbf{Y}$. 

All methods are trained for 50 epochs with the Adam optimizer at a learning rate of $0.01$. For all unrolled methods (AE, TGSR-DAU, and the proposed method), the number of outer layers is set to $L=3$, as performance did not improve significantly beyond three layers.
\blue{The optimization parameters of the unrolled methods, including the regularization weights and step sizes of the proposed method, are not hand-tuned but trained from data (Section~\ref{subsec:unrolling}).
The remaining hyperparameters, such as the learning rate and the numbers of outer layers and inner iterations, were selected empirically, since no significant improvement was observed in preliminary experiments beyond the reported settings.
For the training-free methods, the parameters were selected via grid search on the training data.}
To evaluate each component of the proposed method, we also perform an ablation study. For ablation, we compare \blue{the following} variations of the proposed model:
\begin{itemize}
    \item \textbf{Fixed modality graph}: Inspired by \cite{jiang2021sensor}, we fix the modality graph as a weighted path graph, which only models the smoothness individually per spatial node. The number of trainable parameters is $O(LT)$.
    \item \textbf{Without graph learning}: Inspired by \cite{nagahama2022multimodal}, we compute the twofold graph from the input as a fixed nearest-neighbor graph, and use it across the layers. The number of trainable parameters is $O(L)$.
    \item \textbf{Iterative solver}: The proposed method without unrolling, with the regularization parameters tuned via grid-search to achieve the best MSE for denoising on data with SNR $10$ dB.
    \item \blue{\textbf{Unsigned modality graph}: The sign modulation is disabled by replacing $\mathbf{C}$ with the all-ones matrix; the architecture and parameter count are otherwise identical to the proposed model.}
    \item \blue{\textbf{Pairwise signs}: The polarity is estimated element-wise as $\operatorname{sign}(K_{ij})$, analogous to the pairwise rule in \cite{dittrich_learning_2020}.}
    \item \blue{\textbf{Rank-1 spectral signs}: The polarity is $\mathbf{S}=\mathbf{s}\mathbf{s}^\top$, where $\mathbf{s}$ is the entrywise sign of the principal eigenvector of $\operatorname{sign}(\mathbf{K})$.}
\end{itemize}
\blue{The last three variants differ from the proposed model only in the modality sign pattern.}

\blue{Statistical significance is assessed with two paired tests against the proposed method: a Wilcoxon signed-rank test over paired per-sample masked MSEs (same fold, evaluation sample, and mask), reported per missing rate and pooled over the five rates, and a paired t-test over per-fold mean MSEs. All p-values are Holm--Bonferroni-adjusted across the compared methods within a table.
The outcomes of the pooled Wilcoxon test are reported directly in the captions of Tables~\ref{tab:synth_mn_result}--\ref{tab:ablation}.}

\subsubsection{Results}
Tables~\ref{tab:synth_mn_result}\blue{, \ref{tab:synth_pws_result}, and \ref{tab:synth_pwcs_result}} report the average MSE at 20~dB SNR. We also show the number of \blue{trainable} parameters for each method in the last column.
The proposed method achieves the lowest error \blue{among all competing methods} on \blue{all three} datasets across all missing rates.
\blue{The only exception is its own iterative variant at a missing rate of 40\% on the piecewise smooth dataset, where the difference is marginal (0.033 vs. 0.036).
At 20 dB SNR, the proposed method reduces the average masked MSE relative to the best competing method by 12\% on the matrix normal dataset, 83\% on the piecewise smooth dataset, and 74\% on the signed piecewise smooth dataset.}
Figure~\ref{fig:mse_vs_noise} extends this comparison across SNR levels, confirming that the proposed method consistently outperforms all baselines over the full noise range.

From the ablation, we observe that fixing the modality graph to a path graph yields competitive performance on the piecewise smooth dataset, where a path graph is a reasonable prior approximation. However, the performance degrades substantially when the signal distribution deviates from that assumption for matrix normal dataset (Fig.~\ref{fig:mse_vs_noise}(a)). 
Disabling graph learning leads to a larger drop on the piecewise smooth dataset (Table~\ref{tab:synth_pws_result}), confirming that adaptive graph estimation is essential for signals with complex inter-modal structure. 
The iterative solver without unrolled training achieves competitive accuracy but tends to oversmooth under high noise (Fig.~\ref{fig:mse_vs_noise}), motivating end-to-end learning of the regularization parameters.

\blue{Table~\ref{tab:ablation} summarizes the ablation over the sign configurations, and Table~\ref{tab:sign_estimators} reports the measured edge-sign accuracy of the corresponding sign estimators, where the accuracy of the unsigned estimator equals the fraction of positive edges in the ground truth.
As we can see, the optimal sign configuration, including the unsigned one, depends on the nature of the modality correlations in the given data.
We adopt the spectral estimator in \eqref{eq:spectral_sign} with rank $r=2$ as the default owing to its robustness across the three datasets, while the sign estimation module can easily be replaced according to the application.}

\begin{figure}
    \centering
    \subfigure[Matrix normal]{\includegraphics[width=0.32\linewidth]{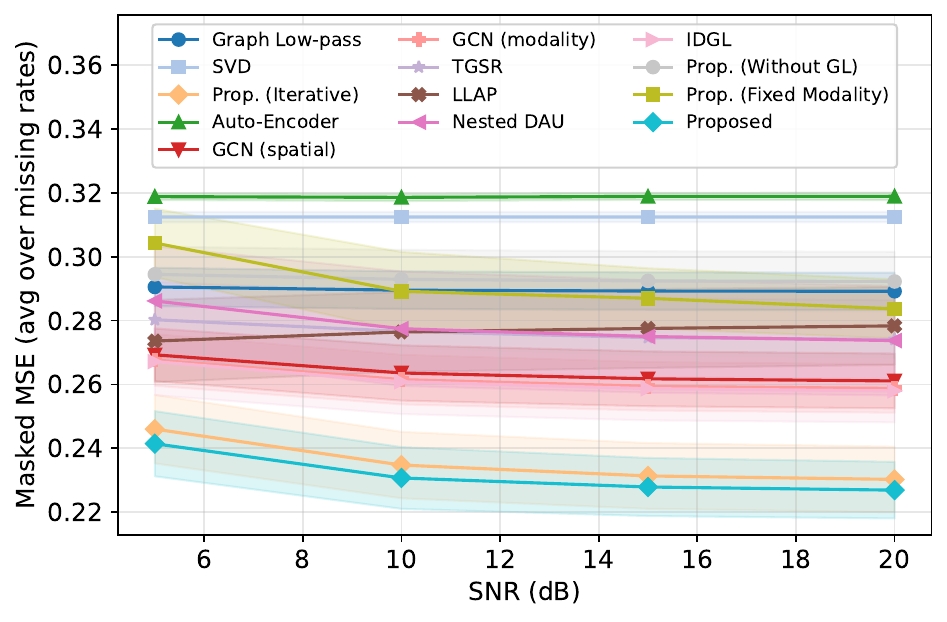}}
    \subfigure[Piecewise smooth]{\includegraphics[width=0.32\linewidth]{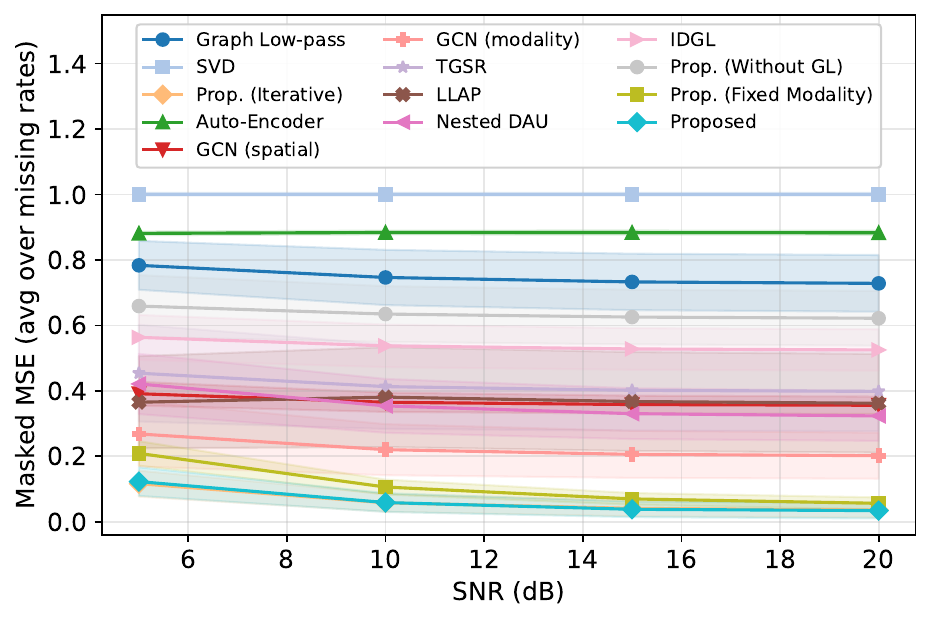}}
    \subfigure[Signed piecewise smooth]{\includegraphics[width=0.32\linewidth]{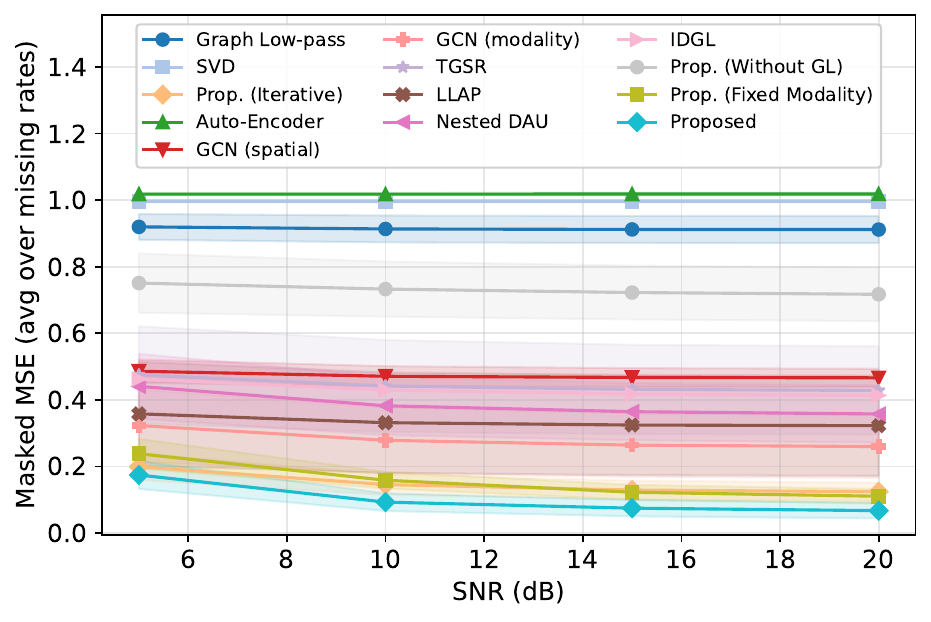}}
    \caption{\blue{Masked MSE as a function of SNR. Values are averaged over all folds and missing rates; shaded regions indicate cross-fold standard deviations.}}
    \label{fig:mse_vs_noise}
\end{figure}

\begin{table}[tb]
\centering
\caption{Average MSE on Matrix Normal Dataset at 20 dB. \blue{All differences from the proposed method are statistically significant ($p<0.05$)}}
\vspace{-8pt}
\label{tab:synth_mn_result}
\resizebox{\linewidth}{!}{%
\begin{tabular}{@{}clcccccc@{}}
\toprule
&Methods\textbackslash Missing rate (\%) & 10 & 20 & 30 & 40 & 50 & Params. \\
\midrule \midrule
\multirow[c]{3}{*}{Model-based}&Graph Low-pass           & $0.281$   & $0.285$   & $0.289$   & $0.293$   & $0.297$  & 1 \\
&SVD                      & $0.310$   & $0.311$   & $0.312$   & $0.314$   & \blue{$0.315$} & 1 \\
&Proposed (Iterative)     & $0.217$   & $0.224$   & $0.228$   & $0.236$   & $0.246$  & 7  \\\midrule
\multirow[c]{3}{*}{Data-driven}&Auto-Encoder             & \blue{$0.318$}   & $0.318$   & $0.319$   & $0.320$   & $0.320$  & 298392  \\
&\blue{GCN (spatial)}     & \blue{$0.252$}   & \blue{$0.254$}   & \blue{$0.259$}   & \blue{$0.266$}   & \blue{$0.275$}  & \blue{31036} \\
&\blue{GCN (modality)}    & \blue{$0.250$}   & \blue{$0.253$}   & \blue{$0.256$}   & \blue{$0.263$}   & \blue{$0.271$}  & \blue{25906} \\\midrule
\multirow[c]{6}{*}{Unrolled}&TGSR                     & $0.256$   & $0.266$   & $0.274$   & $0.283$   & $0.291$ & 18   \\
& LLAP & \blue{$0.260$}   & \blue{$0.271$}   & \blue{$0.279$}   & \blue{$0.288$}   & \blue{$0.295$}  & 18 \\
&\blue{NestDAU}           & \blue{$0.252$}   & \blue{$0.262$}   & \blue{$0.271$}   & \blue{$0.283$}   & \blue{$0.301$}  & \blue{16} \\
&\blue{IDGL}              & \blue{$0.244$}   & \blue{$0.252$}   & \blue{$0.257$}   & \blue{$0.264$}   & \blue{$0.273$}  & \blue{3} \\
&Prop. (Fixed Modality)   & $0.270$   & $0.278$   & $0.284$   & $0.291$   & $0.296$  & 108 \\
&Prop. (Without GL)       & $0.278$   & $0.286$   & $0.293$   & $0.299$   & $0.305$  & 7 \\
&Proposed                 & $\bf0.215$   & $\bf0.221$   & $\bf0.225$   & $\bf0.234$   & $\bf0.240$  & 183  \\
\bottomrule
\end{tabular}%
}
\end{table}

\begin{table}[tb]
\centering
\caption{Average MSE on Piecewise Smooth Dataset at 20 dB. \blue{All differences from the proposed method are statistically significant ($p<0.05$)}}
\vspace{-8pt}
\label{tab:synth_pws_result}
\resizebox{\linewidth}{!}{%
\begin{tabular}{@{}clcccccc@{}}
\toprule
&Methods\textbackslash Missing rate (\%) & 10 & 20 & 30 & 40 & 50 & Params. \\
\midrule \midrule
\multirow[c]{3}{*}{Model-based}&Graph Low-pass           & $0.614$   & $0.658$   & $0.724$   & $0.795$   & $0.851$ & 1  \\
&SVD                      & $1.009$   & $0.994$   & $0.998$   & $1.004$   & $1.001$  & 1 \\
&Proposed (Iterative)     & $0.020$   & $0.020$   & $0.025$   & $\bf0.033$   & \blue{$0.081$} & 7  \\\midrule
\multirow[c]{3}{*}{Data-driven}&Auto-Encoder & $0.892$   & \blue{$0.873$}   & $0.882$   & $0.886$   & $0.882$ & 293892  \\
&\blue{GCN (spatial)}     & \blue{$0.356$}   & \blue{$0.332$}   & \blue{$0.335$}   & \blue{$0.351$}   & \blue{$0.406$}  & \blue{31036} \\
&\blue{GCN (modality)}    & \blue{$0.137$}   & \blue{$0.146$}   & \blue{$0.174$}   & \blue{$0.224$}   & \blue{$0.329$}  & \blue{25906} \\\midrule
\multirow[c]{6}{*}{Unrolled}&TGSR     & $0.249$   & $0.306$   & $0.378$   & $0.467$   & $0.593$  & 18  \\
&LLAP                     & \blue{$0.179$}   & \blue{$0.242$}   & \blue{$0.334$}   & \blue{$0.458$}   & \blue{$0.598$}  & 18 \\
&\blue{NestDAU}           & \blue{$0.239$}   & \blue{$0.264$}   & \blue{$0.303$}   & \blue{$0.359$}   & \blue{$0.453$}  & \blue{16} \\
&\blue{IDGL}              & \blue{$0.448$}   & \blue{$0.474$}   & \blue{$0.515$}   & \blue{$0.565$}   & \blue{$0.624$}  & \blue{3} \\
&Prop. (Fixed Modality)   & $0.038$   & $0.041$   & $0.052$   & $0.063$   & $0.088$ & 108  \\
&Prop. (Without GL)       & $0.514$   & $0.557$   & $0.612$   & $0.676$   & $0.749$  & 7 \\
&Proposed                 & $\bf0.013$   & $\bf0.017$   & $\bf0.024$   & $0.036$   & $\bf0.079$ & 183  \\
\bottomrule
\end{tabular}%
}
\end{table}

\begin{table}[tb]
\centering
\caption{\blue{Average MSE on Signed Piecewise Smooth Dataset at 20 dB. All differences from the proposed method are statistically significant ($p<0.05$)}}
\vspace{-8pt}
\label{tab:synth_pwcs_result}
\resizebox{\linewidth}{!}{%
\begin{tabular}{@{}clcccccc@{}}
\toprule
&Methods\textbackslash Missing rate (\%) & 10 & 20 & 30 & 40 & 50 & Params. \\
\midrule \midrule
\multirow[c]{3}{*}{Model-based}&Graph Low-pass           & $0.849$   & $0.889$   & $0.911$   & $0.946$   & $0.964$  & 1 \\
&SVD                      & $0.992$   & $1.000$   & $0.991$   & $1.002$   & $0.998$ & 1 \\
&Proposed (Iterative)     & $0.092$   & $0.103$   & $0.117$   & $0.140$   & $0.168$  & 7  \\\midrule
\multirow[c]{3}{*}{Data-driven}&Auto-Encoder             & $1.014$   & $1.021$   & $1.014$   & $1.024$   & $1.020$  & 298392  \\
&GCN (spatial)            & $0.465$   & $0.449$   & $0.436$   & $0.465$   & $0.516$ & 31036  \\
&GCN (modality)           & $0.161$   & $0.184$   & $0.228$   & $0.308$   & $0.418$ & 25906  \\\midrule
\multirow[c]{7}{*}{Unrolled}&TGSR                     & $0.261$   & $0.328$   & $0.406$   & $0.514$   & $0.633$ & 18   \\
&LLAP                     & $0.137$   & $0.207$   & $0.291$   & $0.421$   & $0.557$  & 18 \\
&NestDAU                  & $0.258$   & $0.297$   & $0.331$   & $0.399$   & $0.503$  & 16 \\
&IDGL                     & $0.346$   & $0.378$   & $0.405$   & $0.447$   & $0.492$  & 3 \\
&Prop. (Fixed Modality)   & $0.084$   & $0.094$   & $0.106$   & $0.121$   & $0.146$  & 108 \\
&Prop. (Without GL)       & $0.610$   & $0.659$   & $0.706$   & $0.775$   & $0.836$  & 7 \\
&Proposed                 & $\bf0.043$   & $\bf0.050$   & $\bf0.059$   & $\bf0.075$   & $\bf0.109$  & 183  \\
\bottomrule
\end{tabular}%
}
\end{table}

\begin{table}[tb]
\centering
\caption{\blue{Ablation study: masked MSE at 20 dB averaged over missing rates. All differences from the proposed configuration are statistically significant ($p<0.05$); $^\ddagger$ marks significantly better variants}}
\vspace{-8pt}
\label{tab:ablation}
\resizebox{\linewidth}{!}{%
\begin{tabular}{@{}lccc@{}}
\toprule
Variant & Matrix Normal & Piecewise Smooth & Signed Piecewise Smooth \\
\midrule \midrule
Proposed (Iterative)          & $0.230$ & $0.036$ & $0.124$ \\
Prop. (Without GL)            & $0.292$ & $0.622$ & $0.717$ \\
Prop. (Fixed Modality)        & $0.284$ & $0.056$ & $0.110$ \\
Prop. (Unsigned modality)     & $0.228$ & $\bf0.017^\ddagger$ & $0.069$ \\
Prop. (Pairwise signs)        & $\bf0.222^\ddagger$ & $0.033$ & $0.075$ \\
Prop. (Rank-1 spectral signs) & $\bf0.222^\ddagger$ & $0.059$ & $0.090$ \\
Proposed                      & $0.227$ & $0.034$ & $\bf0.067$ \\
\bottomrule
\end{tabular}%
}
\end{table}

\subsubsection{Graph Recovery Evaluation}
\blue{To directly validate the learned graphs, we compare the estimated adjacency matrices at each unrolled layer against the ground-truth graphs (20 dB SNR, 10\% missing rate, 5 folds $\times$ 20 evaluation samples).
Figure~\ref{fig:graph_recovery} reports two metrics computed on the upper triangle of each adjacency matrix: The edge AUC, \ie, the threshold-free ranking of the estimated edge magnitudes $|\hat{W}_{ij}|$ against the true edge support, and the relative adjacency error $\min_a \|a\hat{\mathbf{W}}-\mathbf{W}\|_F/\|\mathbf{W}\|_F$, which accounts for the unidentifiable global scale.
Each trajectory traces the estimates from the first to the final unrolled layer.
The modality graph is compared against LLAP; for the spatial graph, IDGL is also included.}

\blue{We emphasize that the objective in \eqref{eq:original_objective} is not to recover the ground-truth graphs but to find the twofold graph that best restores the signals, as reflected by the restoration results above.
Nevertheless, the proposed method estimates graphs closer to the ground truth than the competing methods in most cases.
An interesting observation is that the recovery metrics of the spatial and modality graphs are in a trade-off relationship across the unrolled layers for almost all datasets.
On the signed piecewise smooth dataset, for instance, the spatial graph improves over the layers, moving toward the upper left in Fig.~\ref{fig:graph_recovery}(c), while the modality graph gradually degrades, moving toward the lower right.
This suggests that the unrolled training splits the dominant filtering axis per layer.
Such a pattern is not observed for LLAP.}


\begin{figure}[tb]
    \centering
    \subfigure[Matrix normal]{\includegraphics[width=0.32\linewidth]{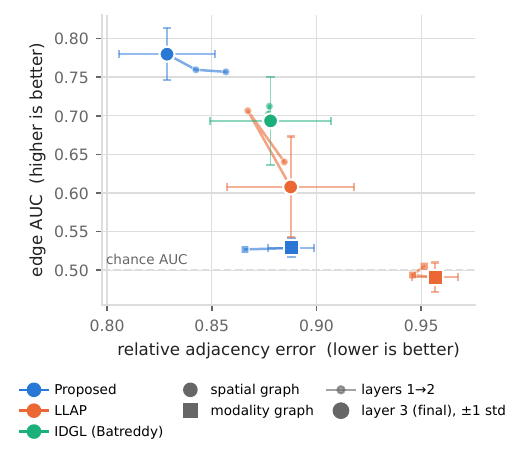}}
    \subfigure[Piecewise smooth]{\includegraphics[width=0.32\linewidth]{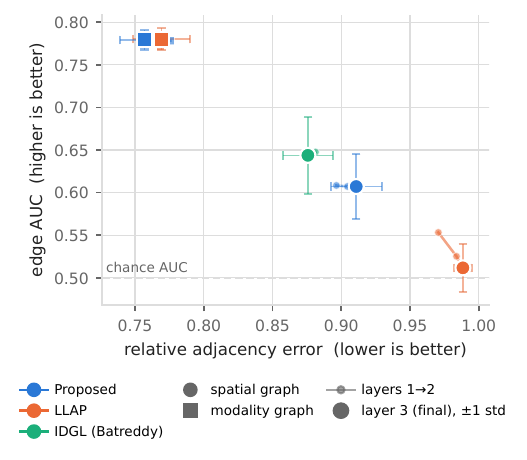}}
    \subfigure[Signed piecewise smooth]{\includegraphics[width=0.32\linewidth]{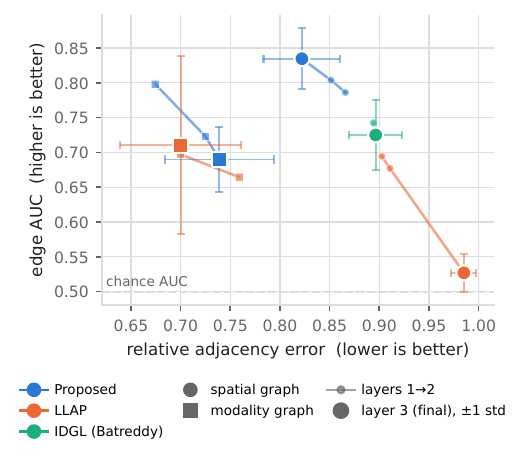}}
    \caption{\blue{Graph recovery at 20 dB SNR and a 10\% missing rate: Edge AUC versus relative adjacency error for the spatial (circles) and modality (squares) graphs. Markers trace the three unrolled layers; the final layer is shown with $\pm1$ standard deviation.}}
    \label{fig:graph_recovery}
\end{figure}

\begin{table}[tb]
\centering
\caption{\blue{Edge-sign accuracy of the sign estimators at the final layer (20 dB, 10\% missing rate)}}
\vspace{-8pt}
\label{tab:sign_estimators}
\resizebox{\linewidth}{!}{%
\begin{tabular}{@{}lccccc@{}}
\toprule
& \multicolumn{2}{c}{Matrix Normal} & \multicolumn{2}{c}{Signed Piecewise Smooth} & Piecewise Smooth \\
\cmidrule(lr){2-3}\cmidrule(lr){4-5}\cmidrule(l){6-6}
Estimator & Sign acc. & Neg. recall & Sign acc. & Neg. recall & Sign acc. \\
\midrule \midrule
Unsigned (all positive) & $0.594$ & $0.000$ & $0.607$ & $0.000$ & $\bf1.000$ \\
Pairwise ($\operatorname{sign}(K_{ij})$) & $0.919 \pm 0.033$ & $0.914$ & $\bf0.996 \pm 0.014$ & $\bf0.996$ & $0.678$ \\
Rank-1 ($\mathbf{s}\mathbf{s}^\top$) & $\bf0.991 \pm 0.018$ & $\bf0.991$ & $0.926 \pm 0.090$ & $0.916$ & $0.770$ \\
Subspace ($r=2$, proposed) & $0.858 \pm 0.025$ & $0.853$ & $0.966 \pm 0.057$ & $0.955$ & $0.753$ \\
\bottomrule
\end{tabular}%
}
\end{table}

\subsection{Evaluation on Real-world Data}
We evaluate the proposed method on interpolation of \blue{two real-world datasets: a Japan meteorological dataset and a Beijing air-quality dataset}.

\noindent\textit{Japan Meteorological Dataset}\footnote{Japan Meteorological Agency, ``AMeDAS,''. Available: jma.go.jp}:

The dataset consists of daily records for four modalities: average temperature, average sea-level air pressure, average humidity, and sunshine duration. The records between the years 2018--2025 were collected from 141 observatories located across Japan. We stack $28$ days of observations per month into a single matrix, resulting in $96$ matrices of dimension $141\times 112$. The value of each modality is standardized to have a mean of zero and a standard deviation of one.

\vspace{3pt}\noindent\textit{Beijing Air-Quality Dataset}\footnote{\blue{UCI Machine Learning Repository, ``Beijing Multi-Site Air-Quality Data''. Available: archive.ics.uci.edu/dataset/501}}:

\blue{The dataset consists of hourly air-quality records \cite{zhang2017cautionary} from 12 monitoring stations in Beijing between March 2013 and February 2017 (1461 days).
Each station reports ten quantities: six pollutant concentrations (PM$_{2.5}$, PM$_{10}$, SO$_2$, NO$_2$, CO, and O$_3$) and four meteorological variables (temperature, pressure, dew point, and wind speed).
This dataset complements the meteorological one in two ways: The modality axis consists of ten distinct physical quantities with no natural ordering, and the inter-modal relations are genuinely mixed in sign---51\% of the inter-variable correlations are negative, \eg, temperature--pressure ($-0.83$) and NO$_2$--O$_3$ ($-0.47$).
One sample is a single hourly snapshot $\mathbf{X}\in\mathbb{R}^{12\times10}$, whose rows are the stations and columns are the measured quantities; to keep the sample count comparable to the meteorological experiment, the records are subsampled to one day per week and four hours per day.
Each variable is standardized using the observed entries of the training folds, and entries genuinely missing in the raw records (1.7\%) are excluded from both training and evaluation through the observation mask.}

\subsubsection{Experimental Setup}
We evaluate two types of missing patterns: Missing completely at random (MCAR) and missing with random spatial outages (MRSO). In the MCAR pattern, each entry is removed independently at random, modeling sensor failures whose occurrence is unrelated to the measured values. In the MRSO pattern, a randomly selected subset of stations experiences a fixed blackout\blue{---a 14-day outage for the meteorological dataset, and a loss of the entire station row for the air-quality dataset---}modeling structured outages in which all modalities at a station are simultaneously unavailable\blue{; such entries can only be recovered through the spatial graph}. Example observations and masks for both patterns are shown in Fig.~\ref{fig:sample_data}.
For each type, we vary the missing rate over $\{10,\,20,\,30,\,40,\,50\}$\%: for MCAR this is the fraction of missing entries, and for MRSO it is the fraction of stations with blackouts.
\blue{For the meteorological dataset, we perform} leaving-one-year-out cross-validation, where each year is treated as a validation set, while the rest of the years are used as training sets, and report the average MSEs.
\blue{For the air-quality dataset, we perform eight-fold cross-validation over contiguous blocks of days (732 training and 104 test samples per fold); the training protocol is identical to that of the meteorological dataset.}

\blue{Unlike the synthetic datasets, the columns of the meteorological dataset are ordered within each modality (28 consecutive days per quantity). We therefore evaluate two variants of IDGL on this dataset: The global-difference variant applies the column-difference prior across all columns as published, whereas the block-difference variant restricts the difference operator to the columns within each modality, so that no penalty is imposed across the boundaries between modality blocks.
For the air-quality dataset, whose columns carry no ordering, only the global-difference variant is applicable.
LLAP-DAU is not reported on the air-quality dataset because an epoch of training LLAP-DAU on the air-quality dataset takes more than 9 hours.}

\begin{figure}
    \centering
    \subfigure[MCAR Sample]{\includegraphics[width=0.23\linewidth]{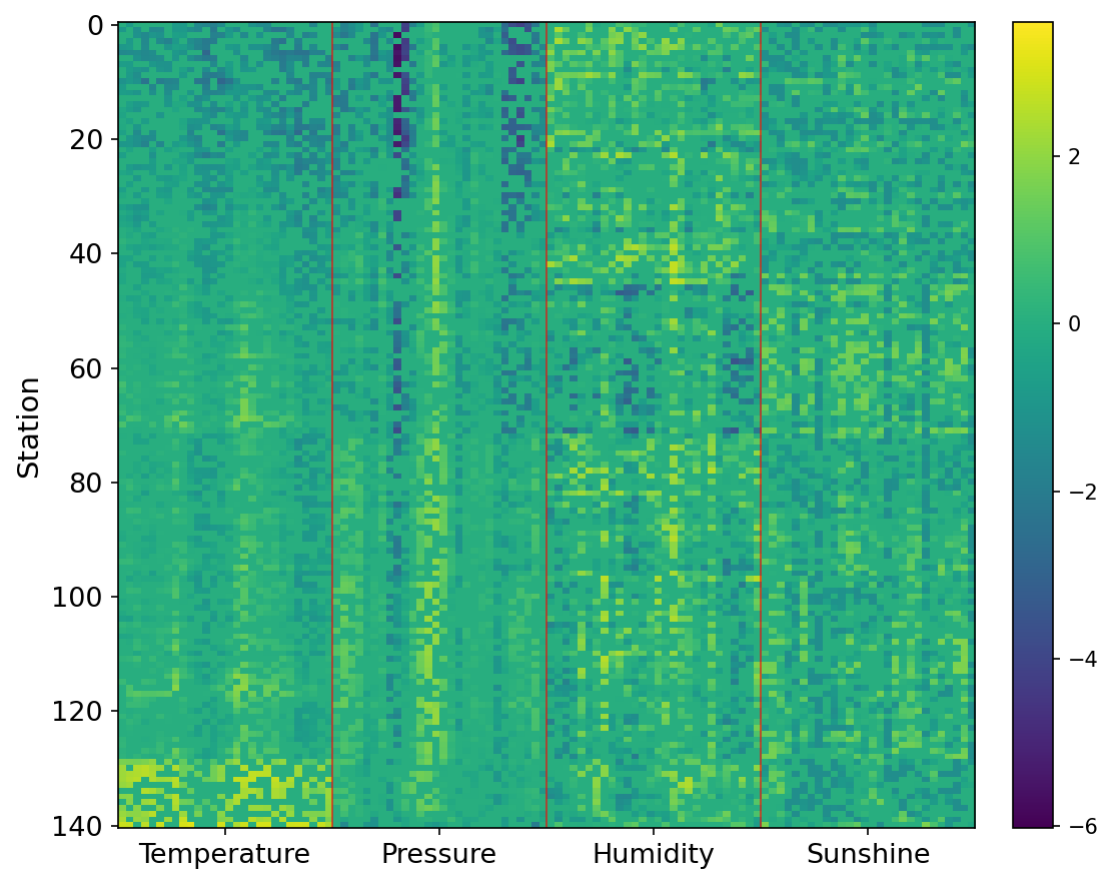}}
    \subfigure[MCAR Mask]{\includegraphics[width=0.23\linewidth]{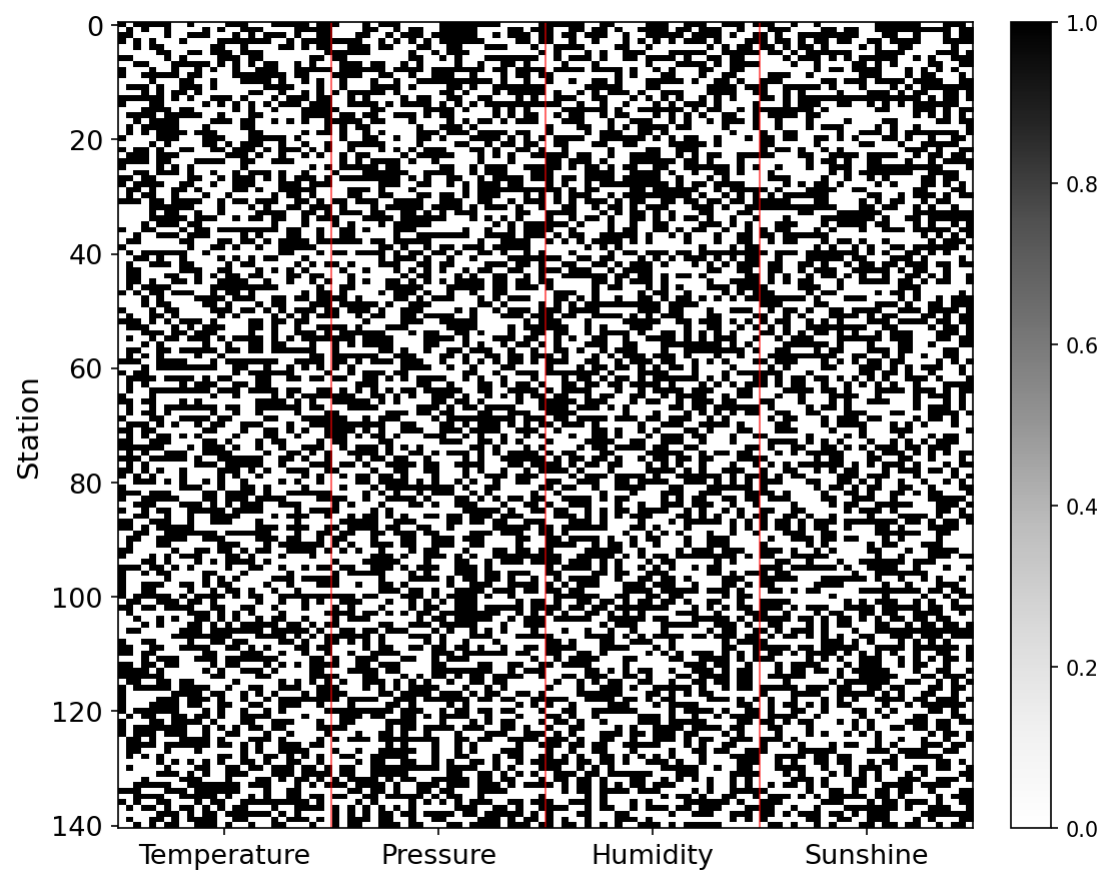}}
    \subfigure[MRSO Sample]{\includegraphics[width=0.23\linewidth]{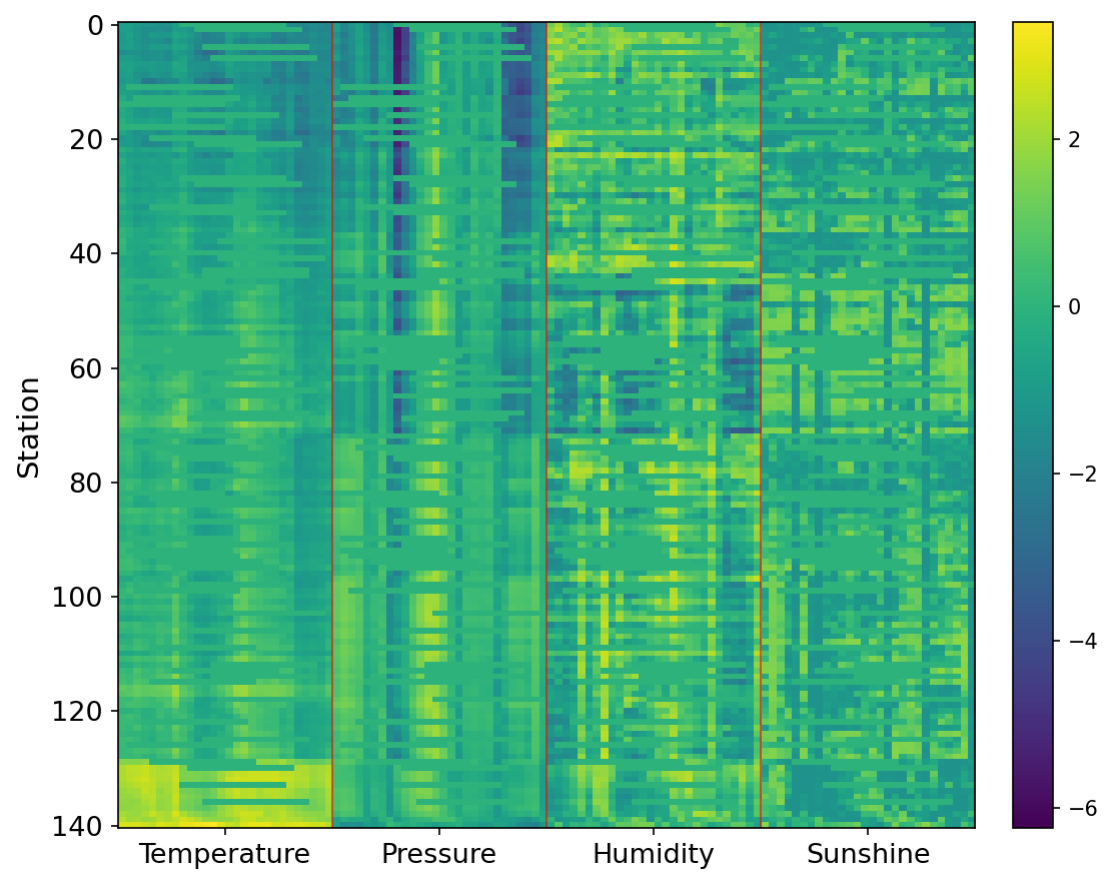}}
    \subfigure[MRSO Mask]{\includegraphics[width=0.23\linewidth]{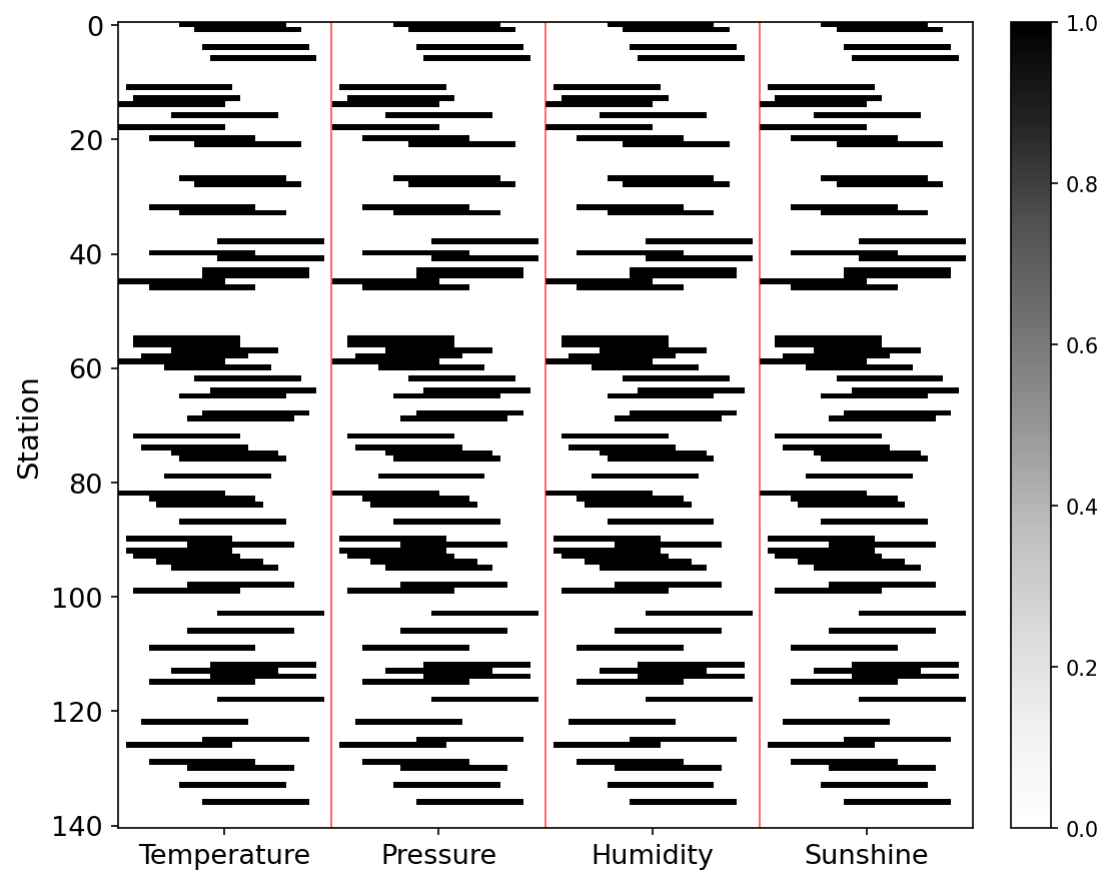}}
    \caption{Example observations and mask matrices of Japan meteorological dataset. Both missing patterns at a missing rate of 50\%.}
    \label{fig:sample_data}
\end{figure}

\subsubsection{\blue{Quantitative Results}}

Figure~\ref{fig:real_mse_vs_drop} shows the masked MSE as a function of missing rate \blue{on the meteorological dataset}.
The proposed method outperforms all baselines across the full missing-rate range for both MCAR and MRSO patterns, with notably smaller cross-fold variance than competing methods, indicating more stable generalization across years.

Figure~\ref{fig:per_modality_mse} reports per-modality \blue{errors}.
While most baselines exhibit a large error in temperature, the proposed method achieves consistently low error across all four modalities (temperature, pressure, humidity, and sunshine duration).
Compared to the fixed-modality variant, the proposed method yields substantially lower errors for humidity and sunshine duration recovery, particularly under the MRSO pattern.
This gap demonstrates the benefit of modeling negative inter-modal correlations through a signed modality graph, which a fixed path graph cannot capture.

\blue{For the air-quality dataset, Fig.~\ref{fig:air_mse_vs_drop} shows the masked MSE as a function of missing rate, and Fig.~\ref{fig:air_per_modality} reports the per-modality errors.
The proposed method achieves the lowest error among all competing methods under both patterns.
IDGL is competitive under MCAR but degenerates under MRSO: It is a purely single-graph method that initializes its graph from the covariance of the observed signal, and under station masking the covariance involving a masked station is zero, so no edges can be learned for those nodes.}

\blue{The comparison between the trained and iterative variants of the proposed method is also informative.
The two variants share the same objective and solver; unrolling adapts only the layer-wise regularization weights and step sizes.
On the matrix normal dataset, the problem is information-limited: The Bayes-optimal error under the generating model is 0.19--0.21, so both variants operate near the floor.
Under station outages on the air-quality dataset, the problem is prior-limited: An entirely unobserved row must be extrapolated through the spatial graph, the layer-wise fidelity--regularization balance has no observations to exploit, and a substantial part of the signal energy (62\% outside the lowest quarter of the modality-graph spectrum) lies in components that the smoothness prior cannot reconstruct for any parameter setting.
Conversely, wherever the observations are informative and the data carry structure the estimator can exploit---the signed piecewise smooth data and MCAR on the real datasets---the trained variant improves on the iterative one by 18--46\%.}

\subsubsection{\blue{Qualitative Results}}

Qualitative reconstructions at a missing rate of \blue{30\%} are shown in Figs.~\ref{fig:output_samples_MCAR} and~\ref{fig:output_samples_MAR}.
The baseline methods tend to oversmooth the signal, losing the temporal dynamics within each modality. The proposed method recovers the temporal structure more faithfully.
\blue{Figures~\ref{fig:output_samples_air_MCAR} and~\ref{fig:output_samples_air_MAR} show the corresponding reconstructions for the air-quality dataset.
Among the unrolled methods, TGSR and NestDAU reproduce only the dominant inter-modality pattern and IDGL leaves visible artifacts on the masked station rows under MRSO, whereas the proposed method restores the masked entries consistently with both the modality pattern and the station-wise variations.}

\begin{figure}
    \centering
    \subfigure[MCAR]{\includegraphics[width=0.48\linewidth]{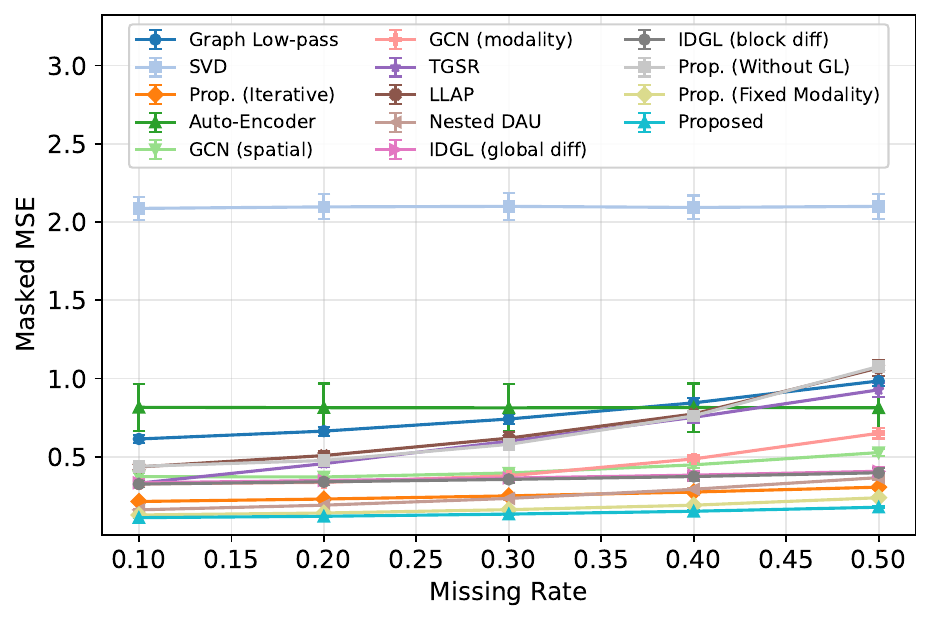}}
    \subfigure[MRSO]{\includegraphics[width=0.48\linewidth]{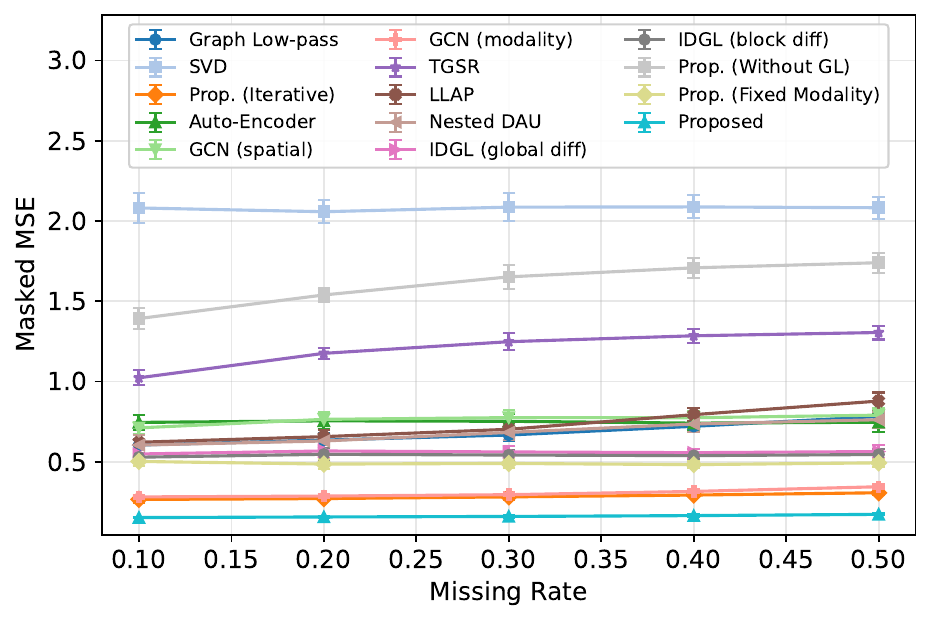}}
    \caption{Masked MSE as a function of missing rate, averaged over all folds. Error bars indicate cross-fold standard deviations.}
    \label{fig:real_mse_vs_drop}
\end{figure}

\begin{figure}
    \centering
    \includegraphics[width=\linewidth]{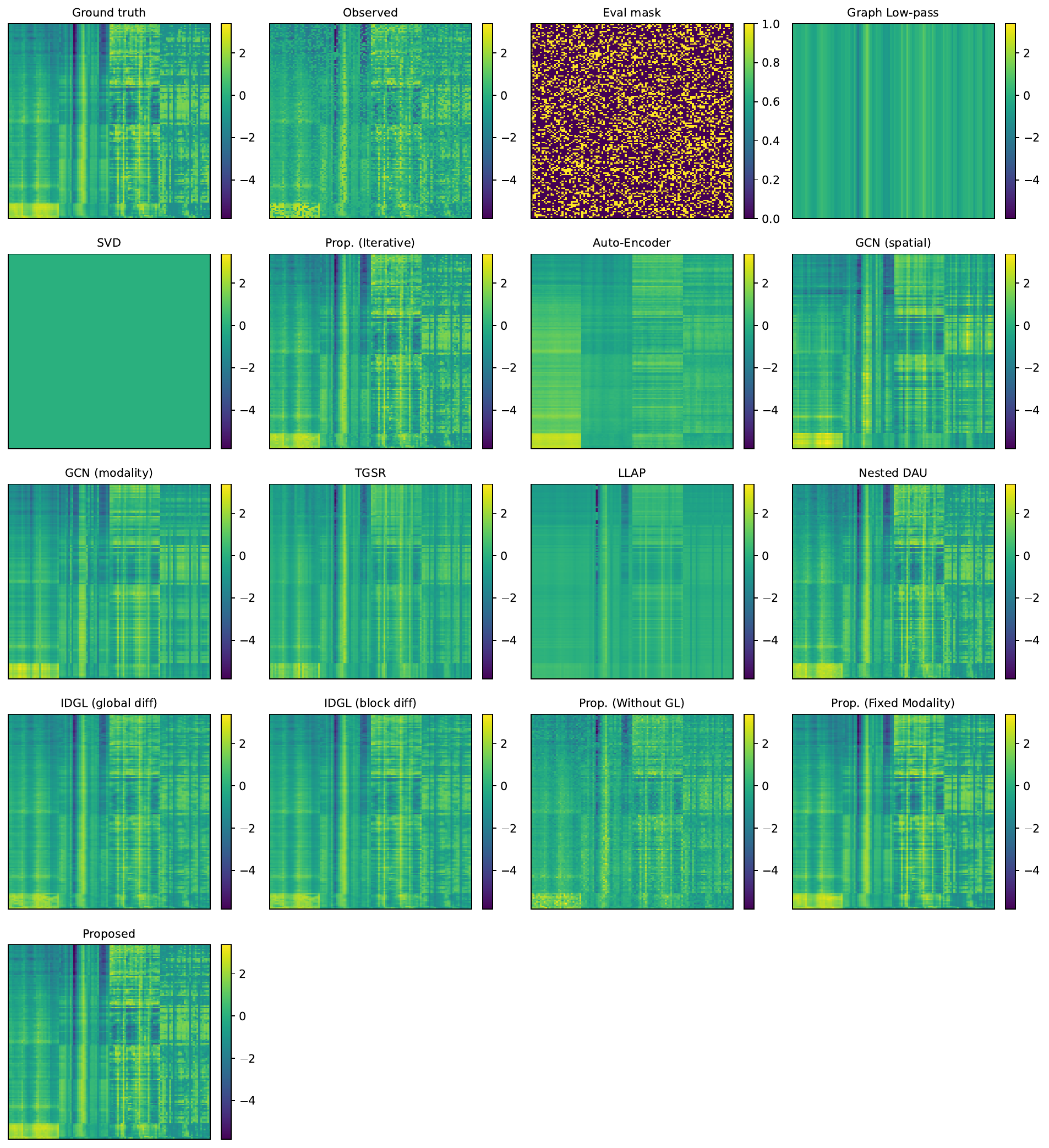}
    \caption{Reconstructed signals under the MCAR pattern at a missing rate of \blue{30\%}.}
    \label{fig:output_samples_MCAR}
\end{figure}

\begin{figure}
    \centering
    \includegraphics[width=\linewidth]{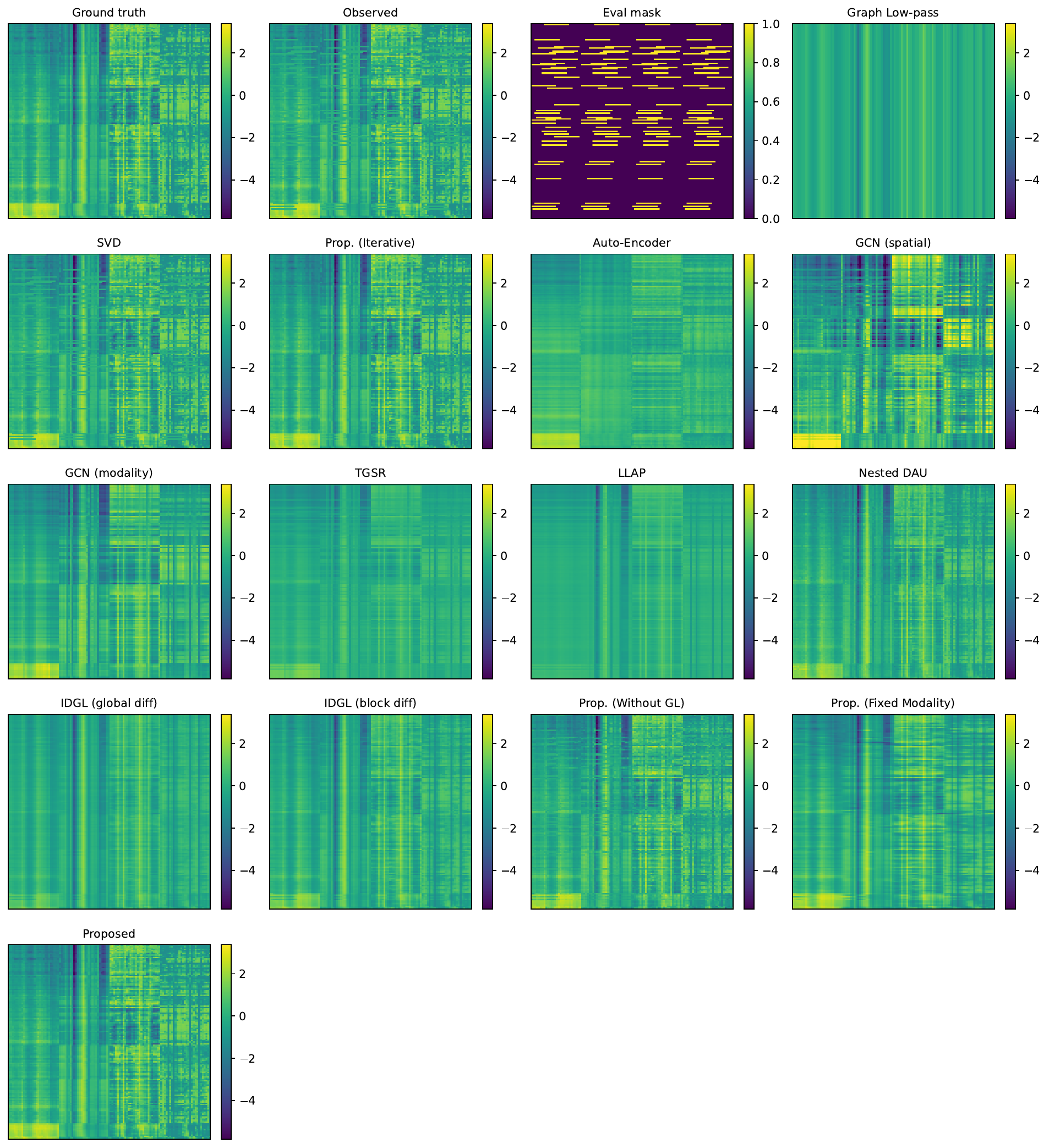}
    \caption{Reconstructed signals under the MRSO pattern at a missing rate of \blue{30\%}.}
    \label{fig:output_samples_MAR}
\end{figure}

\begin{figure}
    \centering
    \includegraphics[width=\linewidth]{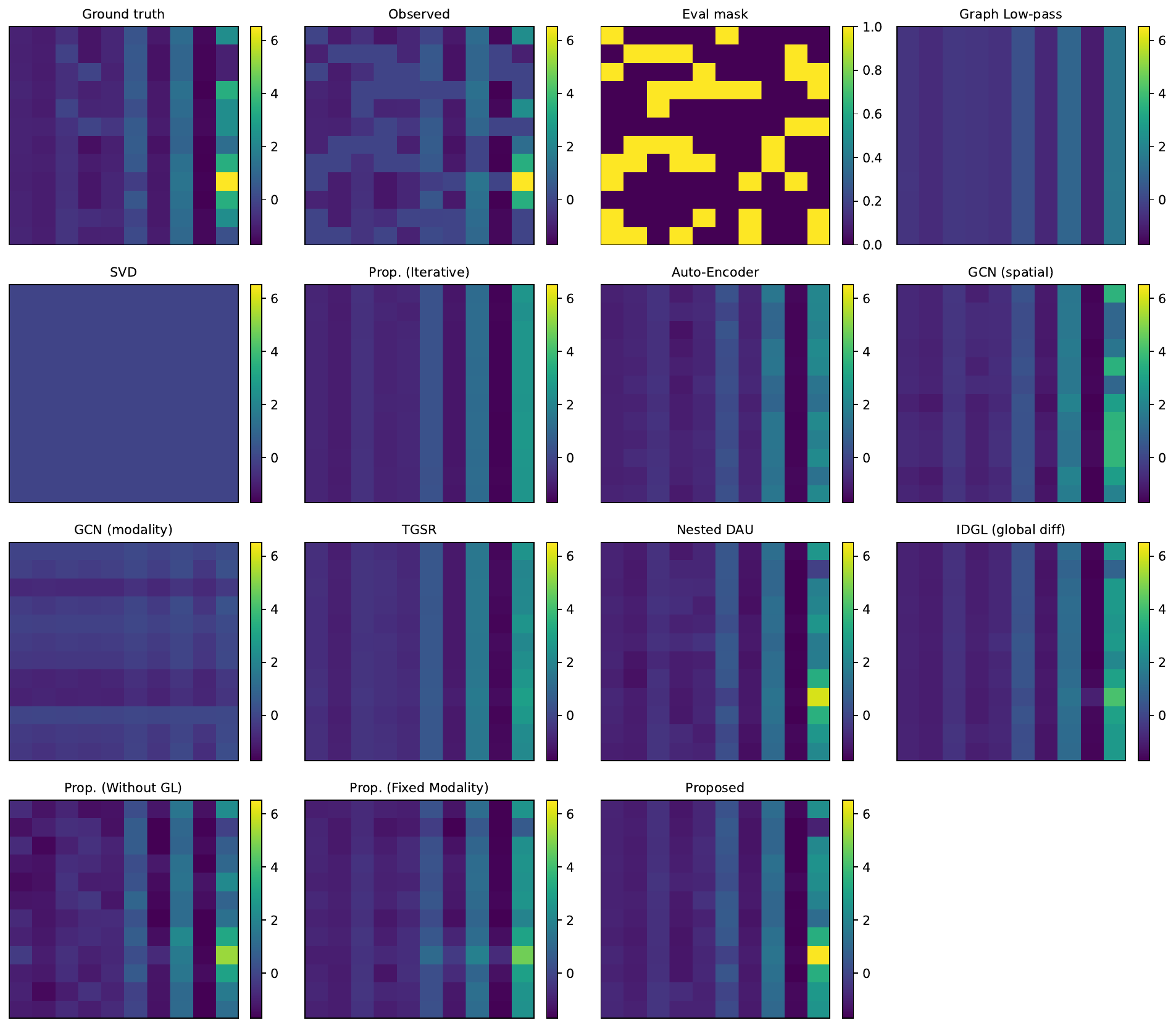}
    \caption{\blue{Reconstructed signals on the air-quality dataset under the MCAR pattern at a missing rate of 30\%.}}
    \label{fig:output_samples_air_MCAR}
\end{figure}

\begin{figure}
    \centering
    \includegraphics[width=\linewidth]{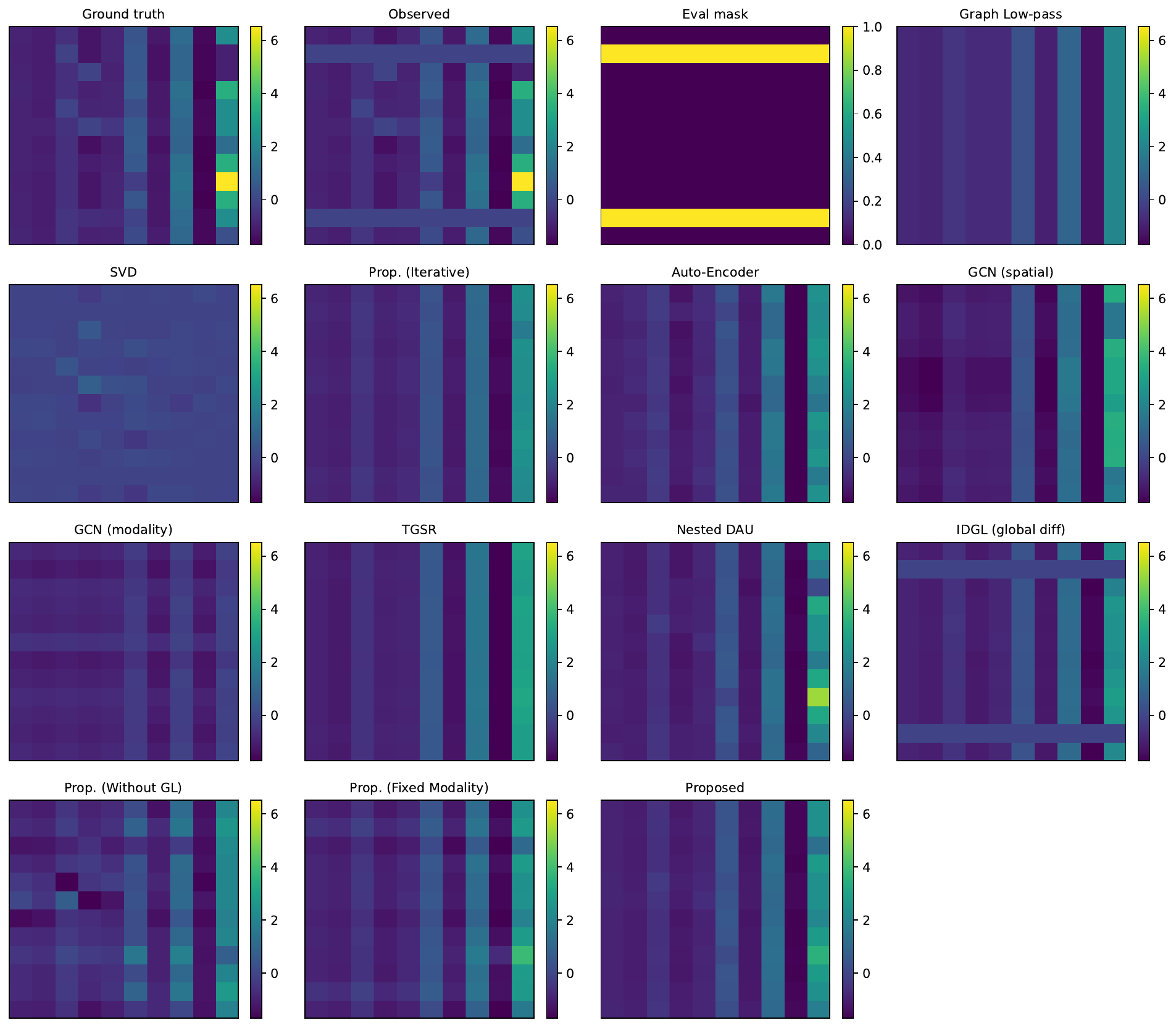}
    \caption{\blue{Reconstructed signals on the air-quality dataset under the MRSO pattern at a missing rate of 30\%.}}
    \label{fig:output_samples_air_MAR}
\end{figure}

\begin{figure}
    \centering
    \subfigure[MCAR]{\includegraphics[width=0.48\linewidth]{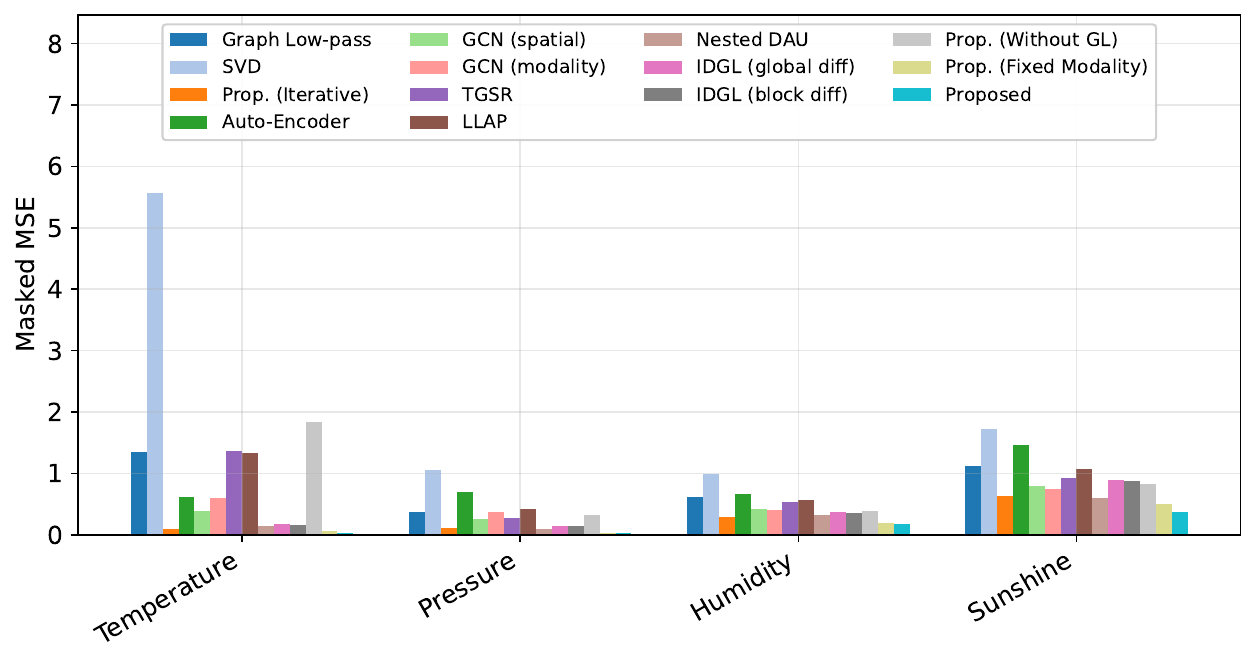}}
    \subfigure[MRSO]{\includegraphics[width=0.48\linewidth]{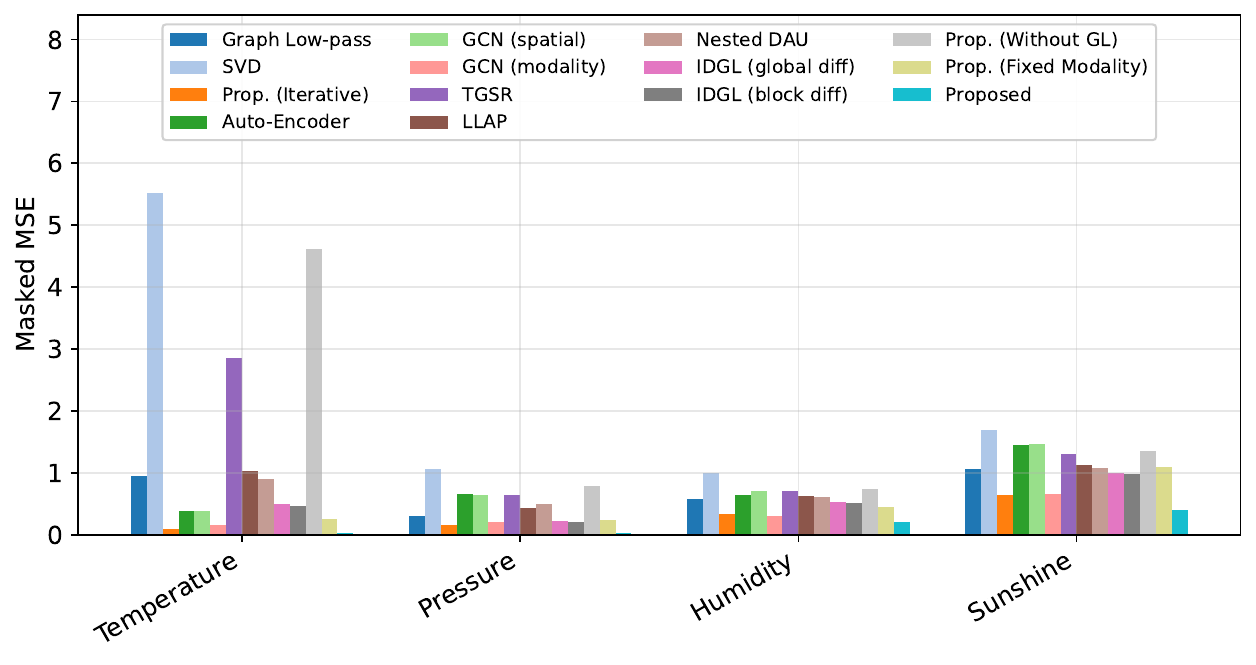}}
    \caption{\blue{Masked MSE} per modality, averaged across all folds and missing rates.}
    \label{fig:per_modality_mse}
\end{figure}

\begin{figure}
    \centering
    \subfigure[MCAR]{\includegraphics[width=0.48\linewidth]{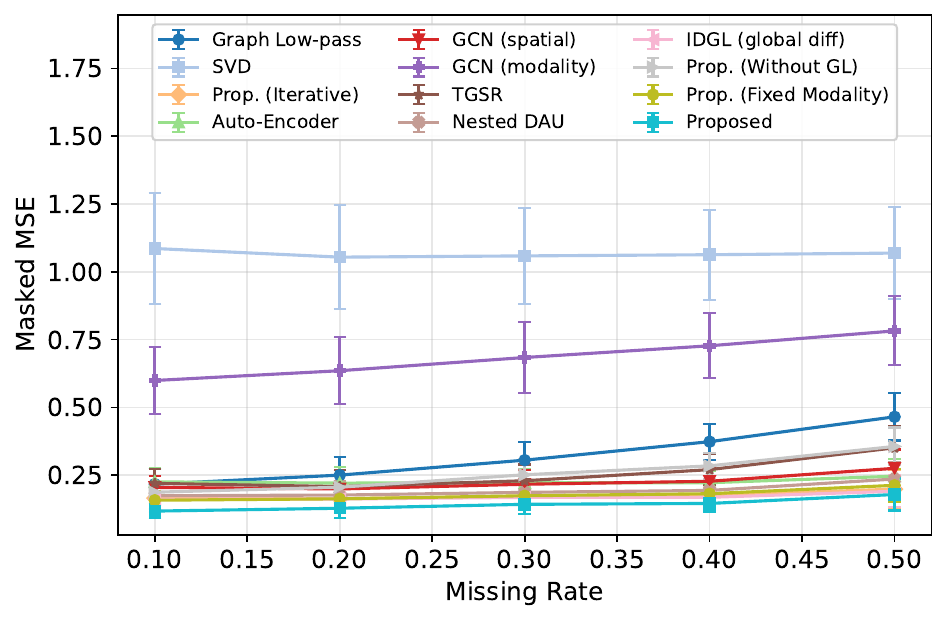}}
    \subfigure[MRSO]{\includegraphics[width=0.48\linewidth]{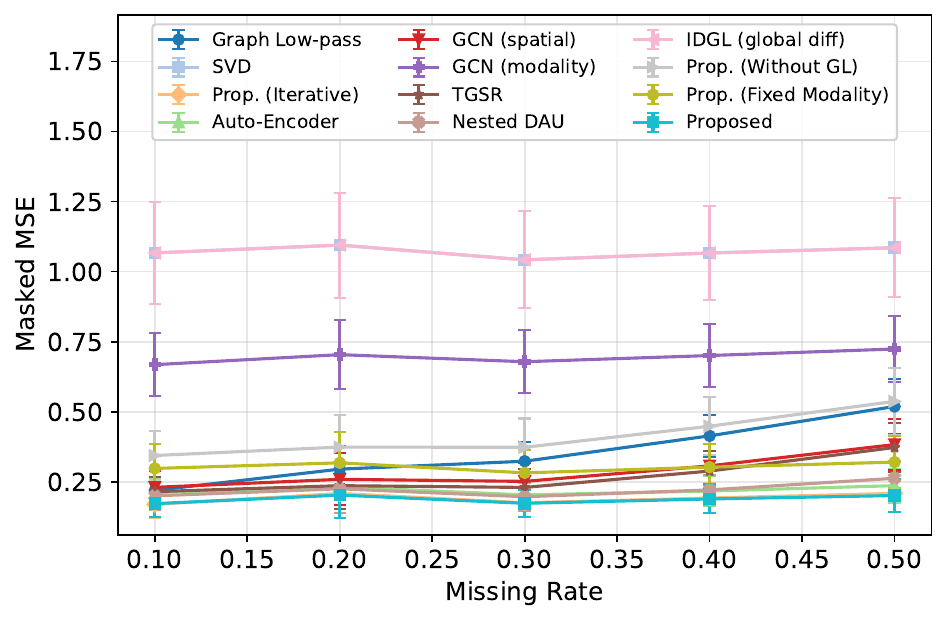}}
    \caption{\blue{Masked MSE on the air-quality dataset as a function of missing rate, averaged over all folds. Error bars indicate cross-fold standard deviations.}}
    \label{fig:air_mse_vs_drop}
\end{figure}

\begin{figure}
    \centering
    \subfigure[MCAR]{\includegraphics[width=0.48\linewidth]{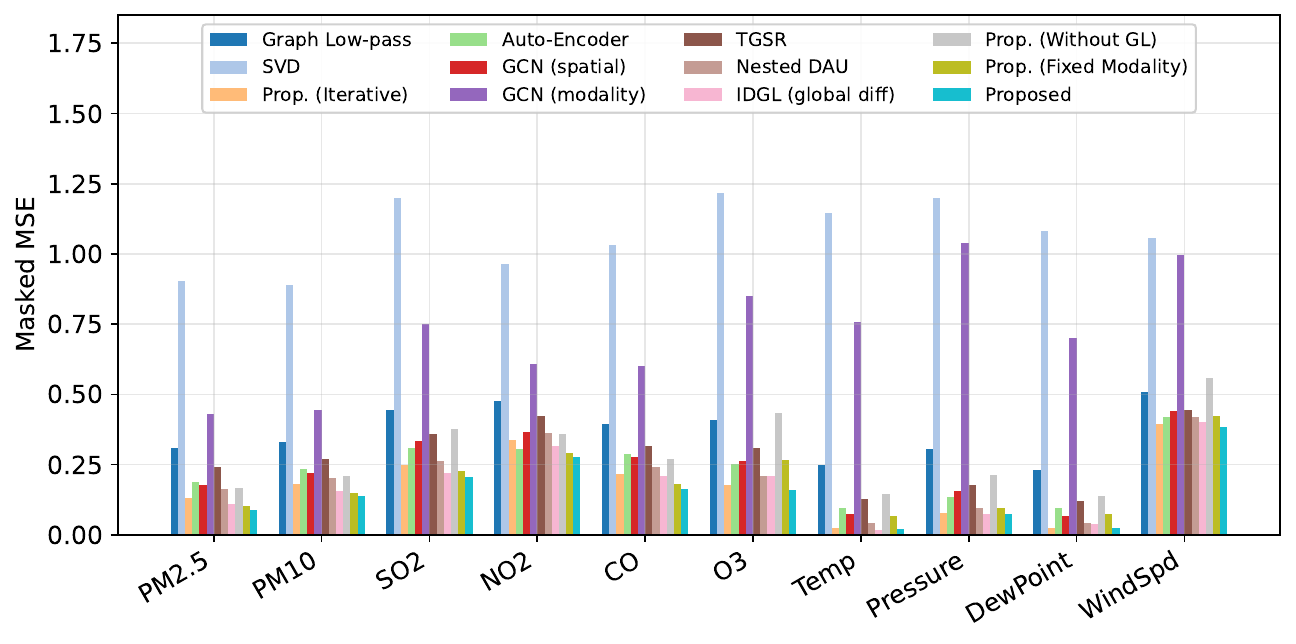}}
    \subfigure[MRSO]{\includegraphics[width=0.48\linewidth]{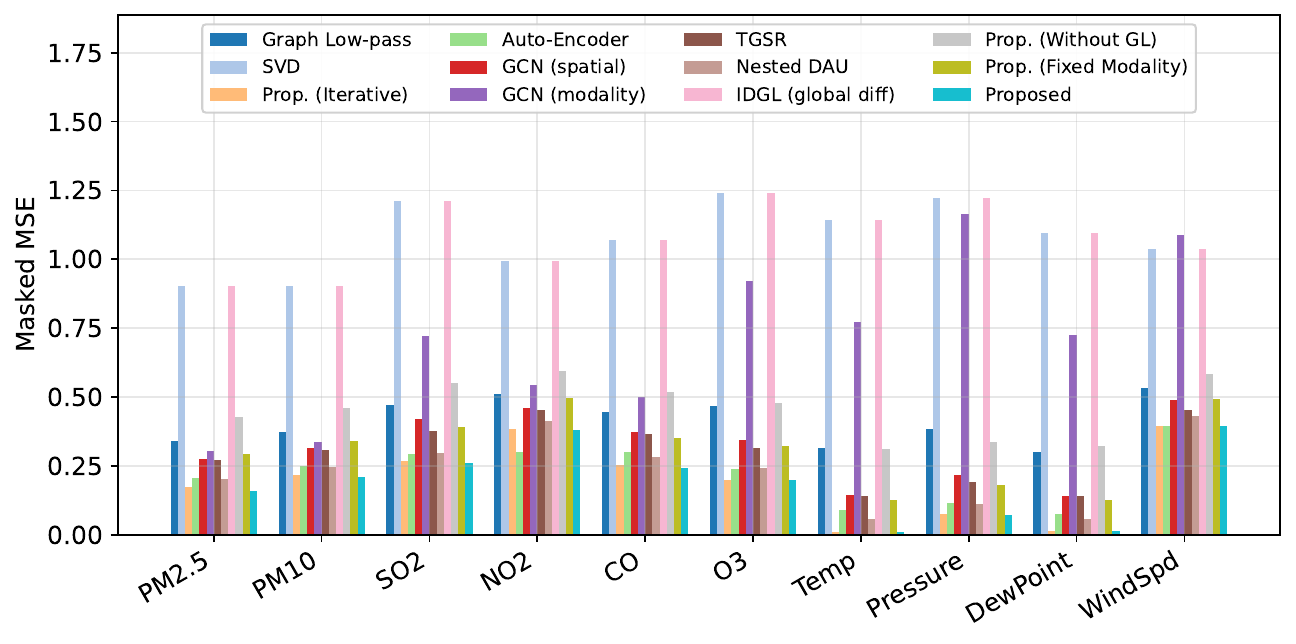}}
    \caption{\blue{Masked MSE per modality on the air-quality dataset, averaged across all folds and missing rates.}}
    \label{fig:air_per_modality}
\end{figure}

\begin{figure}[t]
    \centering
    \subfigure[Spatial graph, MCAR]{\includegraphics[width=0.9\linewidth]{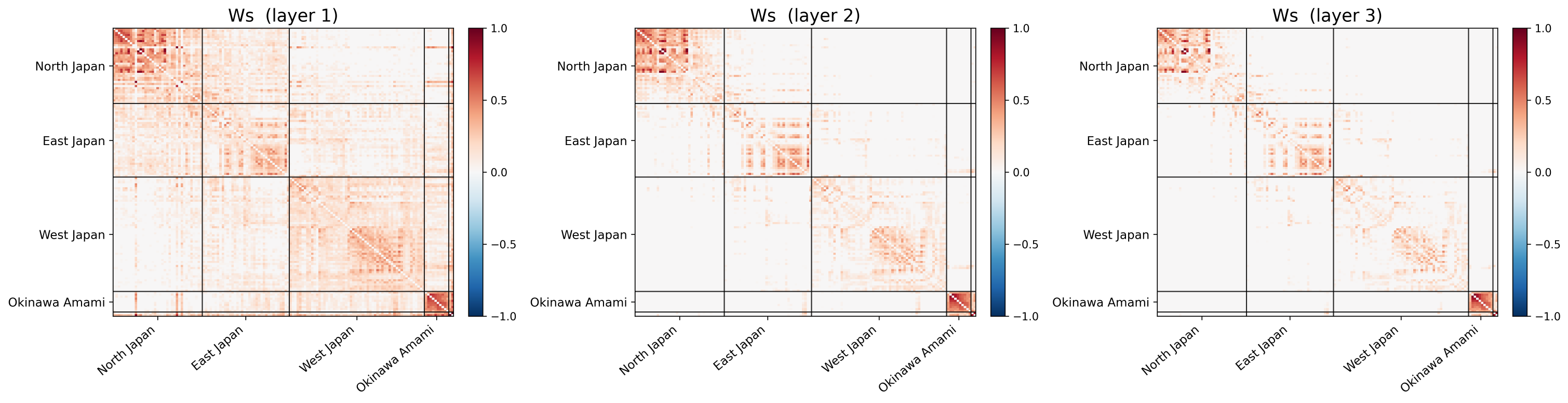}}\\
    \subfigure[Spatial graph, MRSO]{\includegraphics[width=0.9\linewidth]{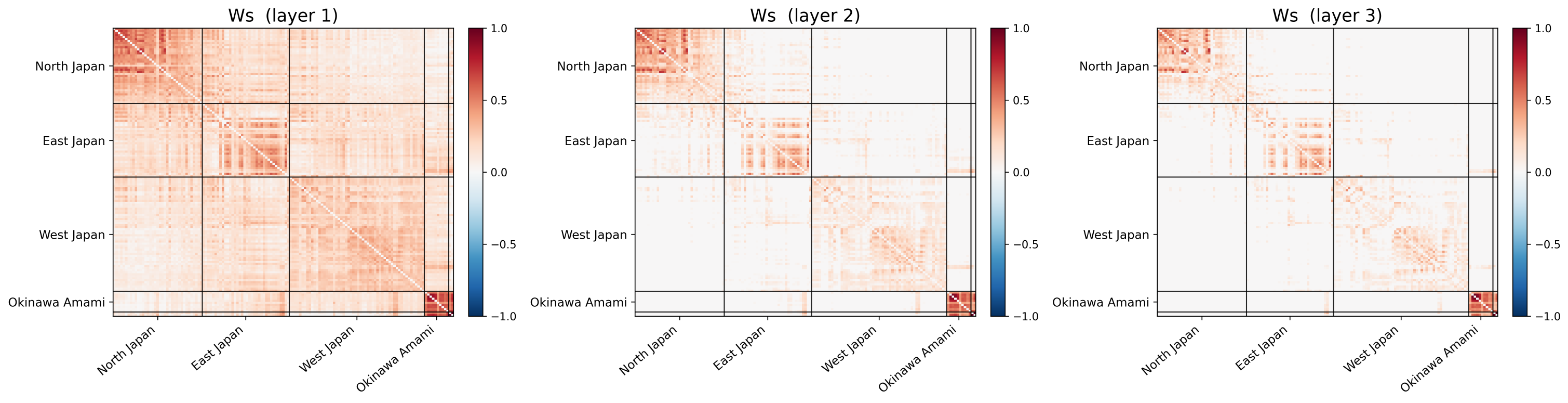}}\\
    \subfigure[Modality graph, MCAR]{\includegraphics[width=0.9\linewidth]{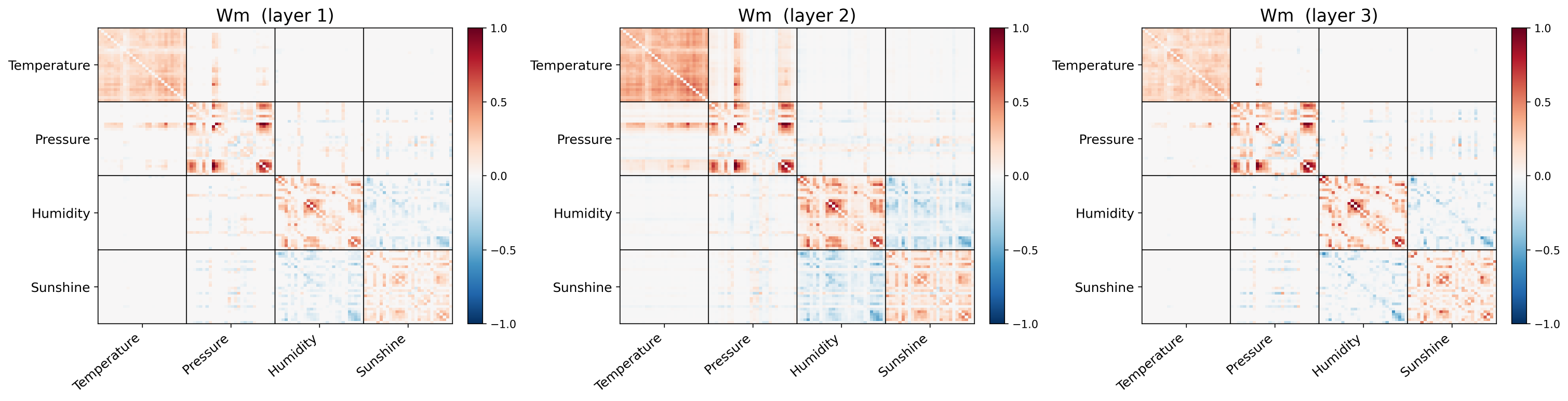}}\\
    \subfigure[Modality graph, MRSO]{\includegraphics[width=0.9\linewidth]{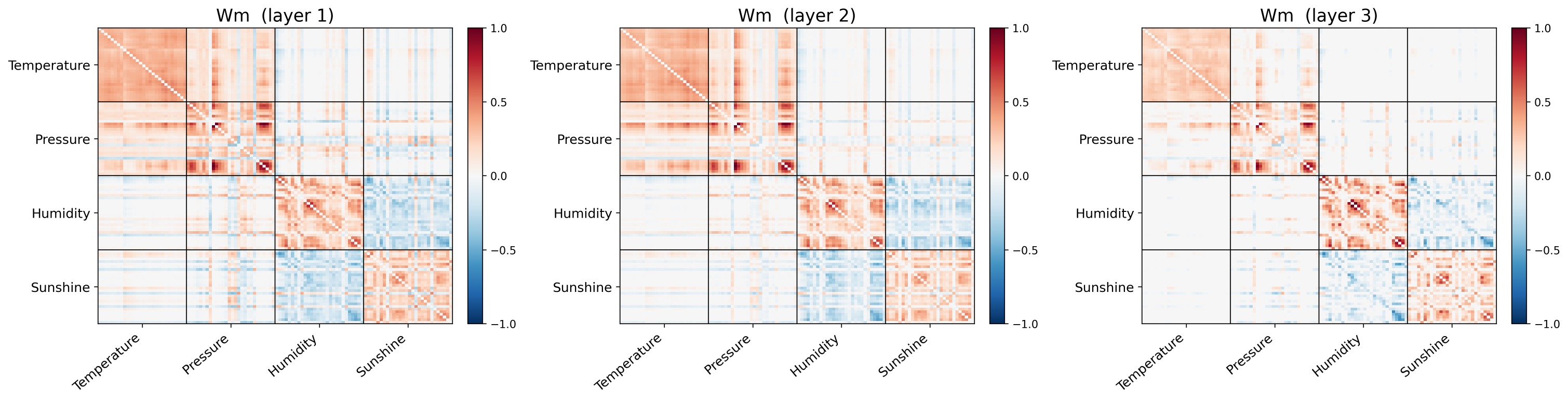}}
    \caption{Learned spatial adjacency matrix $\mathbf{W}_s$ (top) and signed modality adjacency matrix $\bar{\mathbf{W}}_m$ (bottom) for the MCAR (left) and MRSO (right) patterns, from the final outer layer of a representative fold.}
    \label{fig:learned_graphs}
\end{figure}
\begin{figure}[t]
    \centering
    \subfigure[Graph learning parameters, MCAR]{\includegraphics[width=0.9\linewidth]{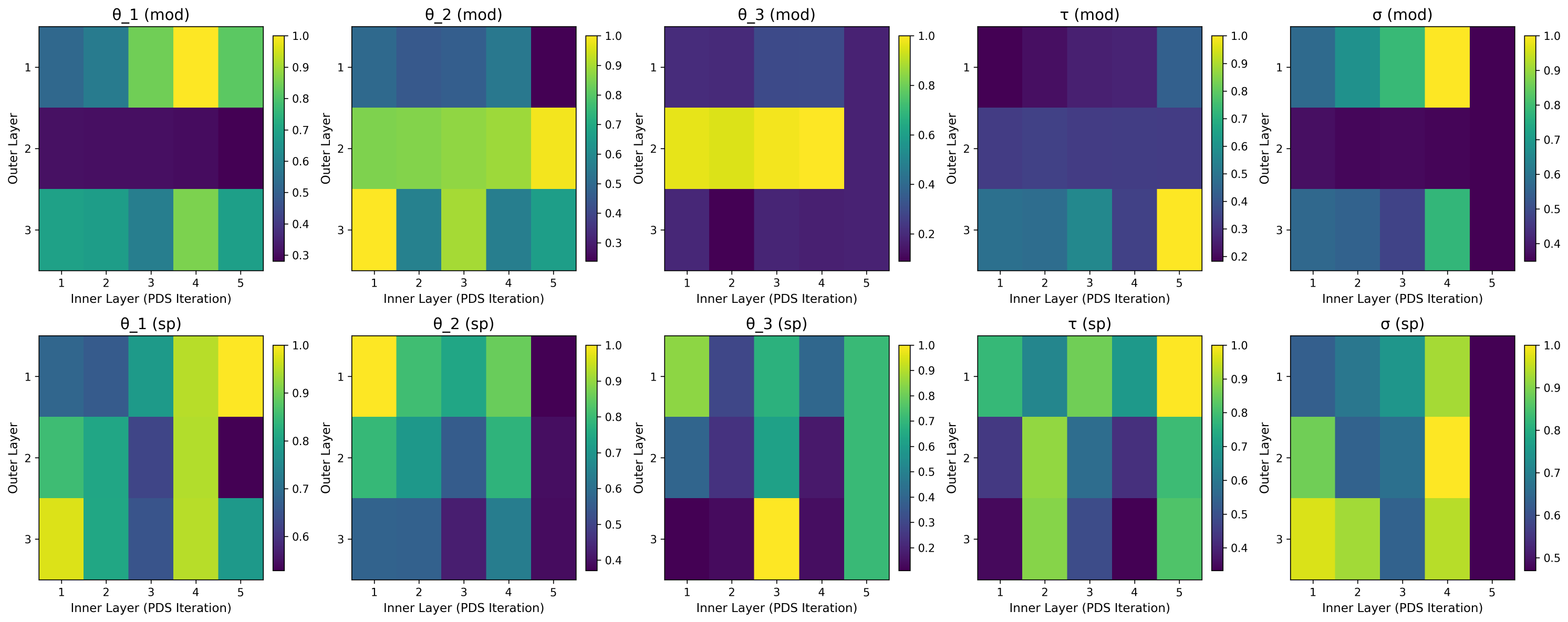}}\\
    \subfigure[Graph learning parameters, MRSO]{\includegraphics[width=0.9\linewidth]{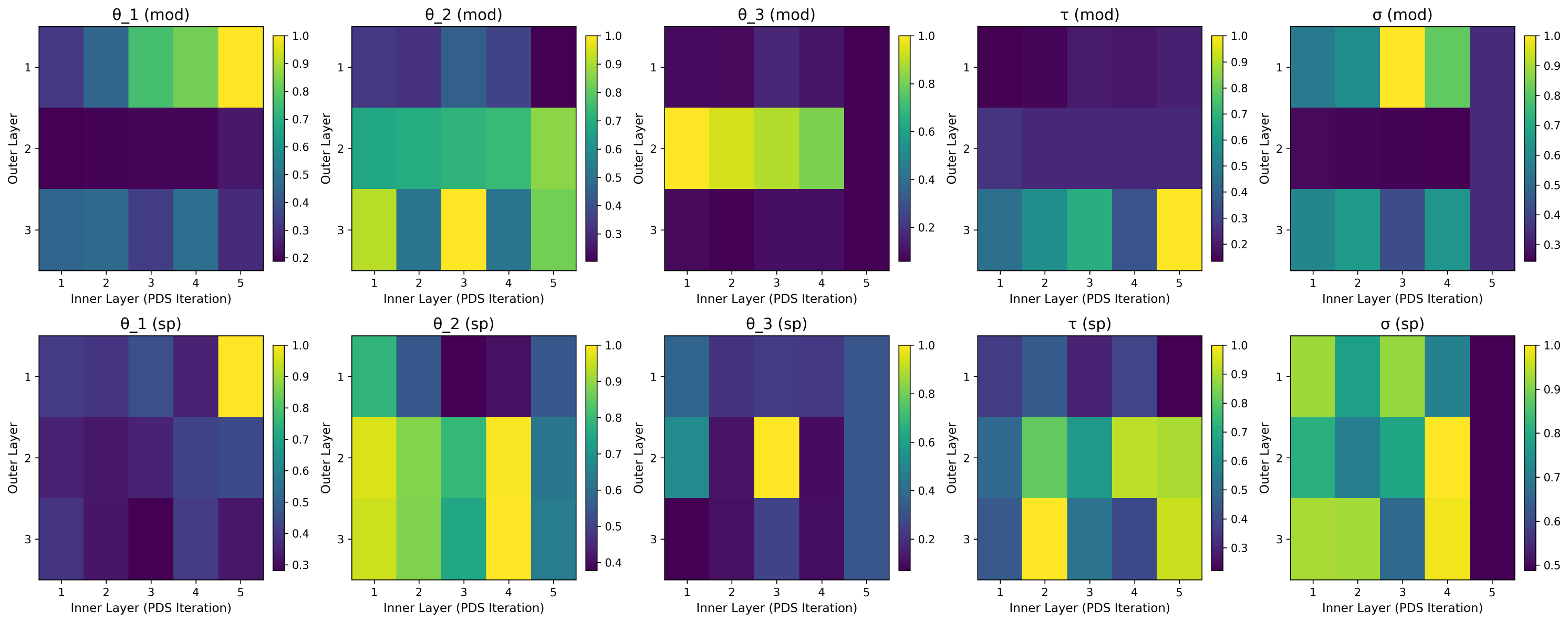}}\\
    \subfigure[Signal reconstruction parameters, MCAR]{\includegraphics[width=0.8\linewidth]{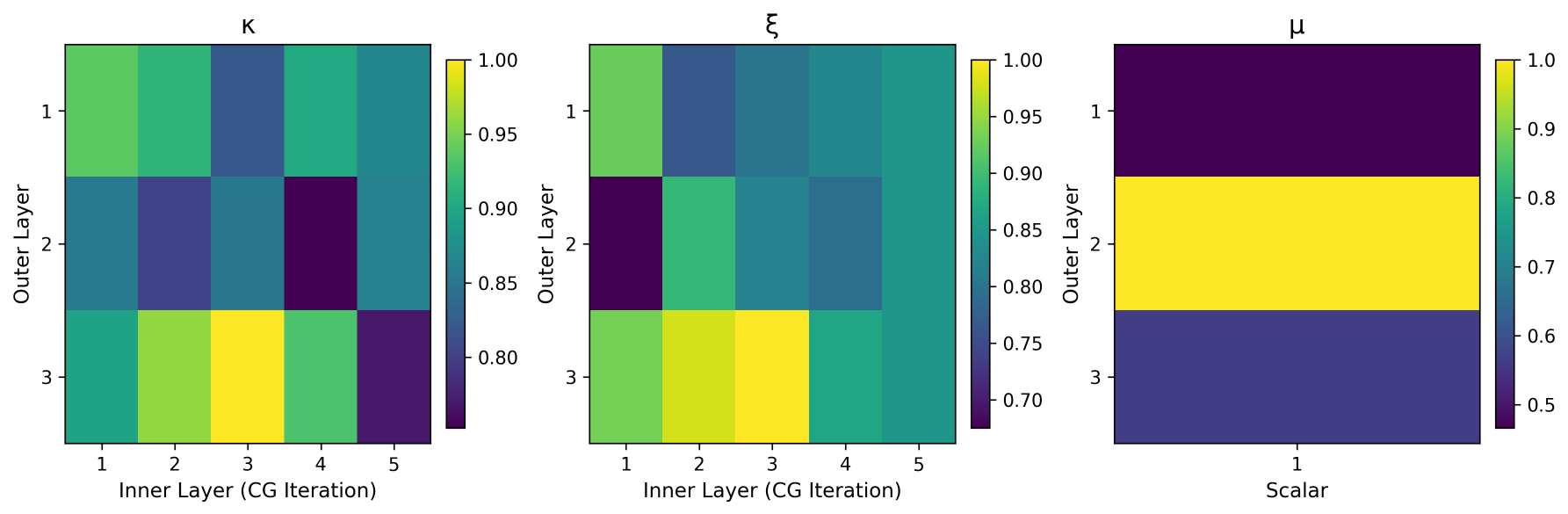}}\\
    \subfigure[Signal reconstruction parameters, MRSO]{\includegraphics[width=0.8\linewidth]{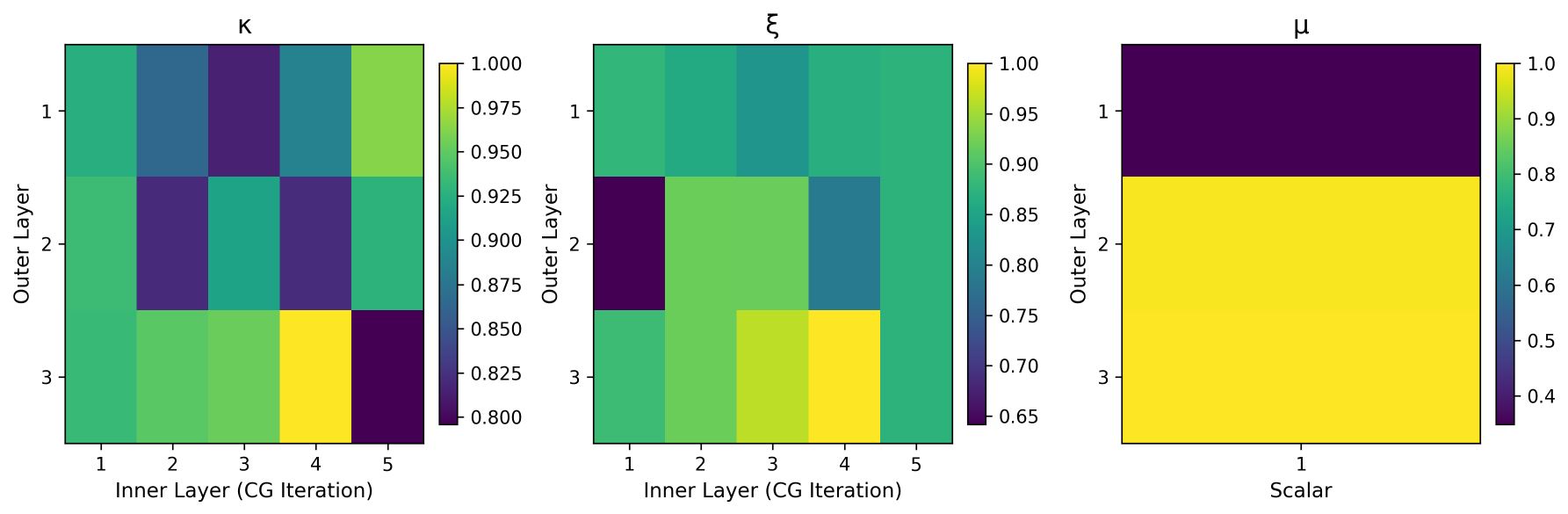}}
    \caption{Learned graph learning parameters (top) and signal reconstruction parameters (bottom) for MCAR (left) and MRSO (right).}
    \label{fig:learned_params}
\end{figure}

\vspace{3pt}\noindent\textit{Learned Graph Structures}:

Figure~\ref{fig:learned_graphs} shows the normalized weights of the learned spatial adjacency matrix $\mathbf{W}_s$ and signed modality adjacency matrix $\bar{\mathbf{W}}_m$ from a representative fold, taken from each layer of the unrolled network.
The spatial graph encodes proximity relations among the 141 meteorological stations partitioned into geographical regions; the modality graph encodes pairwise correlations---including signed ones---among the four observed quantities (temperature, pressure, humidity, and sunshine duration).

From the first to the third layer, the spatial graph becomes more sparse and localized, indicating that the network gathers global information in the early layers and refines local structure in the later layers. 
The modality graph consistently captures strong negative correlations between humidity and sunshine duration, which is a known physical relationship in meteorology.

\vspace{3pt}\noindent\textit{Learned Model Parameters}:

Figure~\ref{fig:learned_params} shows the learned graph-learning regularization parameters ($\theta_1$, $\theta_2$, $\theta_3$, $\tau$, $\sigma$) and signal-reconstruction parameters ($\kappa$, $\xi$, $\mu$) across the three outer layers, for a representative fold.
These parameters are initialized from the iterative solver and refined by end-to-end training; their layer-wise variation reflects the network's adaptation of the regularization strength and solver step sizes to the signal statistics at each stage of the unrolled computation.

From Figs.~\ref{fig:learned_params}(a) and~\ref{fig:learned_params}(b), we see that $\theta_1$ and $\theta_2,\theta_3$ show alternating pattern across the outer layers, which may reflect the network's strategy of alternating between stronger data fitting and stronger sparsity regularization in the graph learning modules. 
The PDHG step sizes $\tau$ and $\sigma$ also vary across layers, suggesting that the network learns to adjust the convergence speed of the inner graph learning iterations at each stage. 
Specifically, primal step size $\tau$ tends to lower values when the structural regularization parameters are higher, which reflects the conservative updates on the graph weights when the regularization is strong.
On the other hand, the dual step size $\sigma$ show a clear pattern of cliff-like decrease near the last inner layer, which may reflect the network's strategy of prioritizing the primal updates over the dual updates in the later stages of graph learning, when the graph estimates are more refined and require more careful adjustments.

From Figs.~\ref{fig:learned_params}(c) and~\ref{fig:learned_params}(d), we see that the twofold-graph regularization weight $\mu$ tends to increase across layers, indicating that the network relies more on the learned graph structure for signal reconstruction in the later stages. 
The CG step sizes $\kappa$ and $\xi$ also show increasing trends across the outer layers. 
This is consistent with the architectural design of the unrolled network, where the later layers can afford more aggressive updates as the graph estimates become more accurate and the signal reconstruction becomes more reliable.

\section{Conclusions}
\label{sec:conclusion}
We proposed an interpretable deep network for jointly restoring multimodal graph signals and learning their underlying twofold graph structure. 
By placing a matrix normal prior on the signal, we derived a MAP objective in which the spatial and modality graph Laplacians serve as precision matrices, and solved it via alternating minimization over three variables: the signal, the unsigned spatial graph, and the signed modality graph. 
Each sub problem is handled by a dedicated differentiable solver---conjugate gradient for the Sylvester-type signal update and PDHG for each graph update---which are unrolled into a feedforward network whose regularization weights and step sizes are learned end-to-end. 
The sign structure of the modality graph is estimated via subspace iteration on a spectral kernel, yielding a continuous relaxation that is compatible with backpropagation.

Experiments on synthetic multimodal graph signals and \blue{two real-world datasets (Japan meteorological and Beijing air-quality data)} confirm that the proposed method outperforms competing baselines in both denoising and interpolation across a range of noise levels and missing-data patterns.
The ablation study shows that signed modality graph learning provides the largest gain when negative inter-modal correlations are present, and that end-to-end training of the unrolled parameters improves robustness under high noise\blue{; the optimal sign configuration depends on the nature of the modality correlations, with the subspace estimator serving as a robust default}.
\blue{The direct graph recovery evaluation shows that the learned graphs are closer to the ground truth than those of competing methods in most cases, although the training objective is restoration rather than graph identification.}
On the \blue{meteorological} dataset, the proposed method also exhibits lower cross-fold variance than the baselines, suggesting stable generalization across different years and missing-data regimes.

\ifdoubleblind\else
\section*{Acknowledgments}
This work was supported in part by JSPS KAKENHI under Grants 24K15047, 25KJ1755, 26H02536 and JST AdCORP under Grant JPMJKB2307.
\fi

\clearpage
\printbibliography

\end{document}